\documentclass{article}
\usepackage{fix-cm, fancybox}
\usepackage{epsf, proof, lscape}
\usepackage{newpxmath, booktabs, multirow, amsmath}
\usepackage[dvipdfmx]{graphicx}
\usepackage{xcolor}

\usepackage{mathdots}

\usepackage{breqn}

\usepackage{verbatim}

\usepackage{tikz}

\usetikzlibrary{intersections, calc, arrows, positioning, arrows.meta}
\usetikzlibrary{arrows,shapes,automata,petri,positioning,calc}\usetikzlibrary {bending}

\usepackage{pxfonts,
proof, latexsym,
epsf,
comment
}

\newtheorem{defn}{Definition}[section]
\newtheorem{theorem}[defn]{Theorem}

\newtheorem{lemma}[defn]{Lemma}

\newtheorem{corollary}[defn]{Corollary}

\newcommand{\QED}{\hspace*{\fill}\vrule height6pt width6pt depth0pt}

\newcommand{\deduc}{
\shortstack{.\\.\hspace{12pt}.\hspace{12pt}.\\.\hspace{8pt}.\hspace{8pt}
.\\
.\hspace{4pt}.\hspace{4pt}.\\...\\.\\.}
}

\newcount\PLv\newcount\PLw\newcount\PLx\newcount\PLy
\newdimen\PLyy\newdimen\PLX \newbox\PLdot \setbox\PLdot\hbox{\tiny.}
\def\scl{.08} 
\def\PLot#1{\PLx`#1\advance\PLx-42\PLy\PLx\PLv\PLx\divide\PLy9%
\PLw\PLy\multiply\PLw9\advance\PLx-\PLw\advance\PLx-4\PLy-\PLy%
\advance\PLy4\PLX=\the\PLx pt\advance\PLyy\the\PLy pt\wd%
\PLdot=\scl\PLX\raise\scl\PLyy\copy\PLdot}
\def\draw#1{\ifx#1\end\let\next=\relax\else\PLot#1%
\let\next=\draw\fi\next}

\def\IA{\hbox{\PLyy=70pt\draw :::;DMV_gqppyyyyyooooxxxnnwvlutkjaWE%
=5-./99:::CCCC:::99/..--544=EEWWaajjjkktttttttVVVVVVVV\end
\hskip7pt}} 
\newbox\IAbox\setbox\IAbox\IA

\begin{document}

\begin{center}\Large
Cut-Elimination for the Bimodal Logic {\sf GR}
\end{center}

\normalsize
\begin{center}

Hirohiko Kushida \\

\vspace{1em}
\large
hkushida@jcga.ac.jp\\
Department of Maritime Safety Technology

Japan Coast Guard Academy
\end{center}

\vspace{1.5em}
\small
{\bf Abstract.}
In this paper, we present a hypersequent calculus
for bimodal logic {\sf GR},
where the two modalities represent the arithmetic
provability predicates of G\"odel and Rosser, respectively.
We prove the cut-elimination theorem for the calculus.

\normalsize
\section{Introduction}
According to G\"odel's first incompleteness theorem,
there exists an arithmetical
sentence $A$
such that neither $A$ nor $\neg A$
is provable
in a formal arithmetic system, e.g., Peano Arithmetic, {\sf PA}.
In the construction of such a sentence $A$, a unary predicate referring to provability is used, and G\"odel assumed the condition
of what is called $\omega$-consistency of the system
in the theorem.
Rosser offered
another provability predicate
to construct such a sentence
so that
the condition of the theorem
is weakened to just the consistency
of the system.
While not all of the so-called
derivability conditions
are satisfied,
as Kreisel \cite{kreisel}
proved,
the consistency assertion on the system itself,
constructed from Rosser's provability predicate,
is actually provable in the system,
that is,
G\"odel's second incompleteness theorem does not hold with the
Rosser's provability predicate.

The study of the G\"odelean provability predicate
via modal logic has been a major field of research in mathematical logic
since Solovay's landmark work
\cite{solovay}
established the axiomatic system of
modal logic {\sf GL}
with its completeness theorems.
\footnote{
Solovay himself called the system {\sf G}.
}
We can refer to a survey paper \cite{ab}
and textbooks \cite{boolos, smo}
for the history and basic results of {\sf GL}.
The modal logical study for the provability predicate of Rosser
was initiated by Guaspari and
Solovay \cite{gs}
with its generalizations,
and it is continued by the authors: 
Visser \cite{visser},
Shavrukov \cite{shav1991},
Sidon \cite{sidon},
Kurahashi \cite{kurahashi2020, kurahashi2021},
Kogure and Kurahashi
\cite{kogure-kurahashi2025}.

In particular,
a specific bimodal logic called {\sf GR}
was proposed in Shavrukov \cite{shav1991}
to analyze the interaction of two distinct modalities
standing for the provability of G\"odel and Rosser, respectively.
Then, Sidon\cite{sidon} proved the Craig interpolation theorem
for the system {\sf GR}.
Recently, the Lyndon interpolation property was shown to hold for {\sf GR}
in Kogure and Kurahashi
\cite{kogure-kurahashi2025}.
However, there has been no decent proof system for {\sf GR} so far;
in the current paper, we offer such a proof system
for which the cut-elimination theorem holds.

Although the current paper is self-contained,
our presentation of hypersequent calculus for {\sf GR}
is based on the proof-theoretic framework
investigated in previous papers \cite{kushida2024, kushida2025},
where
a hypersequent consists of modalized and non-modalized sequents
and certain inference rules on modality work as certain interactions
on the two sorts of sequents.
Furthermore,
we developed a proof-transformation method of the direct cut-elimination,
that is, the cut-elimination without extension to {\it mix},
we call top-down,
for a wide range of modal logics
in \cite{kushida2024, kushida2025};
it is applied for the current case of {\sf GR}.
As a simple application of the cut-elimination theorem,
we also verify of the conservability result of {\sf GR} over
{\sf GL}.

\section{The Bimodal Logic {\sf GR}}
In this section, we offer the definition of syntax and proof systems for
the system {\sf GR}.
Let us begin with the grammar that produces the sentences of {\sf GR}.

\begin{center}
$A := p| \bot | \neg A| A\wedge A| \blacksquare A| \Box A$
\end{center}

\noindent
We also use the other propositional connectives and the other modality $\Diamond$,
which can be defined in terms of the connectives $\wedge, \neg$ and the modality $\Box$.
The modalities $\Box$ and $\blacksquare$ signify
the provability predicates of G\"odel and Rosser, respectively.
First, we review the axiomatic system of the bi-modal logic {\sf GR}.
It consists of the
following axioms and inference rules.

\vspace{1em}
\fbox{The Axiomatic System for {\sf GR}}

\vspace{.5em}
Axioms:

\vspace{.5em}
\hspace{1em}
1. Axioms of Propositional Logic

\vspace{.5em}\hspace{1em}
2. $\Box (A \supset B)\supset (\Box A \supset \Box B)$

\vspace{.5em}\hspace{1em}
3. $\Box (\Box A \supset A)\supset \Box A $

\vspace{.5em}\hspace{1em}
4. $\blacksquare A\supset \Box A $

\vspace{.5em}\hspace{1em}
5. $\Box A \supset \Box \blacksquare A$

\vspace{.5em}\hspace{1em}
6. $\Box A \supset (\Box\bot \vee \blacksquare A)$

\vspace{.5em}\hspace{1em}
7. $\Diamond \blacksquare A \supset \Diamond A $

\vspace{1em}
Inference Rules:

\vspace{.5em}\hspace{1em}
1. $A, A\supset B \vdash B$

\vspace{.5em}\hspace{1em}
2. $A \vdash \Box A$

\vspace{.5em}\hspace{1em}
3. $\Box A \vdash A$

\vspace{1em}

Next, we propose our hypersequent calculus for {\sf GR}
equipped with
two sorts of sequents.
They have the form $\Gamma \to \Delta$ and $\Gamma \Rightarrow \Delta$,
which we call $\to$- and $\Rightarrow$-sequents, respectively.
Here and below, the capital Greek letters $\Gamma, \Delta, \ldots$
denote 
multisets
 of formulas of {\sf GR}.
Their intended meaning is provided
by the following formula image.

\vspace{1em}
$f(\Gamma \to \Delta)=\bigwedge\Gamma \supset \bigvee \Delta$;
$f(\Gamma \Rightarrow \Delta)=\Box(\bigwedge\Gamma \supset \bigvee \Delta)$.

\vspace{1em}
Sequents are denoted by: $S, T, U, \ldots, S_1, S_2, \ldots$.
 $S^{\to}$ ($S^{\Rightarrow}$, respectively) 
 means that $S$ is a $\to$-sequent
 ($\Rightarrow$-sequent, respectively).

Then, the hypersequent is defined to be
the form $S_1 |S_2 |\cdots | S_n$,
which we identify with any permutation $S_{i_1} |S_{i_2} |\cdots | S_{i_n}$.
The intended meaning of it is provided as follows.

\vspace{1em}

$f( S_1 |S_2 |\cdots | S_n)= f(S_1) \vee f(S_2) \vee\cdots \vee f(S_n)$.

\vspace{1em}

Hypersequents are denoted by: ${\cal H}, {\cal I}, {\cal J}, \ldots, {\cal H}_1, {\cal H}_2, \ldots$.
 ${\cal H}^{\to}$ (${\cal H}^{\Rightarrow}$, respectively) 
 means that ${\cal H}$ is made up of $\to$-sequents
 ($\Rightarrow$-sequents, respectively).

\vspace{1em}
Now we present our formulation of
{\sf GR} via the tool of
hypersequent.

\vspace{1em}
\fbox{The Hypersequent Calculus for {\sf GR}}

\vspace{.5em}
$\bullet$ Initial Sequents:
\hspace{2em}
$
A \to
A \hspace{4em}
\bot \to$


\vspace{1em}
$\bullet$ Propositional Inference Rules:

\begin{center}

$
\infer[\wedge:l]{
{\cal H}|
A\wedge B, \Gamma \gg \Delta
}
{ {\cal H}|
A, \Gamma \gg \Delta}
\hspace{2em}
\infer[\wedge:l]{
{\cal H}|
A \wedge B, \Gamma \gg \Delta
}
{
{\cal H}|
B, \Gamma \gg \Delta
}$

\vspace{1em}
$
\infer[\wedge:r]{
{\cal H}| {\cal I}|
\Gamma, \Pi \gg \Delta, \Theta, A \wedge B
}
{ {\cal H}|
\Gamma \gg \Delta, A &
{\cal I}|
\Pi \gg \Theta, B}$

\vspace{1em}
$
\infer[\neg:l]{{\cal H} |
\neg A , \Gamma \gg \Delta
}
{ {\cal H}| \Gamma \gg \Delta, A }
$
\hspace{2em}
$
\infer[\neg:r]{
{\cal H}| \Gamma \gg \Delta, \neg A
}
{{\cal H}| A, \Gamma \gg \Delta}$
\end{center}

\vspace{1em}
$\bullet$ Internal Structural Inference Rules:

\begin{center}
$
\infer[ic]
{ {\cal H} |
A, \Gamma
\gg \Delta
}
{
{\cal H} | A, A, \Gamma \gg \Delta
}
$ \hspace{1em}
$
\infer[ic]
{ {\cal H} |
\Gamma
\gg \Delta, A
}
{{\cal H} |
\Gamma \gg \Delta, A, A
}
$

\vspace{1em}

$
\infer[iw]
{ {\cal H} |
A, \Gamma
\gg \Delta
}
{
{\cal H}| \Gamma \gg \Delta
}
$
\hspace{1em}
$
\infer[iw]
{{\cal H}|
\Gamma
\gg \Delta, A
}
{ {\cal H}|
\Gamma \gg \Delta
}
$

\vspace{1em}

$
\infer[cut]
{ {\cal H}|{\cal I} | \Gamma, \Pi
\gg \Delta, \Theta
}
{ {\cal H} |\Gamma \gg \Delta, A
&
{\cal I} | A,
\Pi
\gg \Theta
}$

\end{center}


$\bullet$ External Structural Inference Rules

\begin{center}
$
\infer[ew]
{{\cal H}|
S 
}
{ {\cal H}
}
$
\hspace{2em}
$\infer[split]{ {\cal H}|
\Gamma \to \Delta |
\Theta \to \Pi
}
{
{\cal H}|
\Gamma,
\Theta \to \Delta, \Pi
}$

\vspace{1em}
$\infer[merge]{
{\cal H}|
\Gamma,
\Theta \gg \Delta, \Pi
}{ {\cal H}|
\Gamma \gg \Delta |
\Theta \gg \Pi
}$
\hspace{2em}
$\infer[\bot]{
{\cal H}|
\Gamma
\Rightarrow \Delta
}{ {\cal H}|
\Gamma \Rightarrow \Delta |
\to
}$

\end{center}

$\bullet$ Inference Rules for Modality

\begin{center}
$
\infer[K ]
{{\cal H}|\Box A \rightarrow | \Gamma \Rightarrow \Delta
}
{{\cal H}|A, \Gamma \Rightarrow \Delta
}$ \hspace{2em}
$
\infer[4]
{{\cal H}|\Box A \rightarrow | \Gamma \Rightarrow \Delta
}
{{\cal H}|\Box A, \Gamma \Rightarrow \Delta
}$
\hspace{2em}
$\infer[\Box :r]
{{\cal H}| \to \Box A
}
{{\cal H}| \Box A \Rightarrow A
}
$

\vspace{1em}
$
\infer[K^\blacksquare]
{{\cal H}|\blacksquare A \rightarrow | \Gamma \Rightarrow \Delta
}
{{\cal H}| A, \Gamma \Rightarrow \Delta
}$
\hspace{2em}
$
\infer[4^
\blacksquare]
{{\cal H}|\blacksquare A \rightarrow | \Gamma \Rightarrow \Delta
}
{{\cal H}|\blacksquare A, \Gamma \Rightarrow \Delta
}$
\hspace{2em}
$
\infer[
\blacksquare:r_1]
{{\cal H}| \to \blacksquare A
| \Rightarrow
}
{{\cal H}|\Box A \Rightarrow A
}$

\vspace{1em}
$\infer[\blacksquare:l]
{{\cal H}|
\blacksquare A\Rightarrow
}
{{\cal H}|
A\Rightarrow
}$
\hspace{2em}
$
\infer[
\blacksquare:r_2]
{{\cal H}| \Rightarrow \blacksquare A
}
{{\cal H}|\Rightarrow A
}$




\end{center}

$\bullet$ Modal Structural Inference Rules

\begin{center}

$
\infer[\Rightarrow]
{S^\Rightarrow
}
{S^\rightarrow
}$
\hspace{1em}
\hspace{1em}
$
\infer[\to]
{ S^\rightarrow
}
{ S^\Rightarrow
}$


\end{center}

Here, all occurrences of $\gg$ in the description of an inference rule
uniformly denote either $\Rightarrow$ or $\rightarrow$.
The names of $K^\blacksquare$ and $4^\blacksquare$
come just from
the similarity
to forms of the rules $K$ and $4$.
The modality $\blacksquare$ does not
satisfy the $K$- nor $4$-axiom.
For two sequents $S=\Gamma \gg \Delta$
and $T=\Pi \gg \Theta$,
$S\cdot T$ is defined to be $\Gamma, \Pi \gg \Delta, \Theta$
and for a hypersequent ${\cal H}=S_1|S_2|\cdots | S_n$, 
$\cdot{\cal H}$ is defined to be $(\cdots (S_1\cdot S_2)\cdot S_3 \cdots) 
\cdot S_n$.


The upper and lower sequents of an inference rule
are naturally defined.
For example,
the upper and lower sequents of $K^\blacksquare$
are, respectively,
$A, \Gamma \Rightarrow \Delta$
and $\blacksquare A\to$.
We set the upper and lower sequent of
the rule $\bot$
to be $\to$ and $\Gamma \Rightarrow \Delta$, respectively,
and the turnstile $\to$ of the upper sequent
is interpreted to disappear.
The upper and lower hypersequents
of an inference rule
are trivially defined
as the whole hypersequents
including the upper and lower
sequents, respectively.

The turnstile $\Rightarrow$ of the upper sequent of $\blacksquare:r_1$
and $\Box:r$
is naturally interpreted to disappear
and
the turnstile $\Rightarrow$ of the lower sequent of $\blacksquare:r_1$
is interpreted to be newly introduced.

Thus, generally in a proof, the turnstile $\Rightarrow$
is introduced by the rules of:

\vspace{1em}
$ew$, $\Rightarrow$, $\blacksquare:r_1$,

\vspace{1em}
\noindent
while
$\to$ is introduced by:

\vspace{1em}
the initial sequent,
$ew$,
$K$, $4$, $K^\blacksquare$ $4^{\blacksquare}$.

\vspace{1em}
Furthermore,
generally in a proof,
the turnstile $\to$ changes
to $\Rightarrow$ solely by the $\Rightarrow$-rule,
while
$\Rightarrow$ changes to
$\to$ by the rules of:

\vspace{1em}
$\to$, $\Box:r$, $\blacksquare:r_1$.

\vspace{1em}
The formula $A$ is called the {\it cut formula} in the definition of $cut$.
The formula $\Box A$
is called the {\it diagonal formula}
in the upper sequents of
the definitions of
$\Box:r$ and $\blacksquare:r_1$.

\subsection{Equivalence to the axiomatic system {\sf GR}}
Here we show the equivalence of
the two systems for
{\sf GR} as exposed above.
Let
${\sf GR}_H$ denote
the hypersequent calculus
and
${\sf GR}_{AX}$ denote
the axiomatic system.

\begin{theorem}
${\sf GR}_H$ and ${\sf GR}_{AX}$ are equivalent with respect to
derivability.
\end{theorem}

\noindent {\bf Proof.}
We omit the propositional logic fragment.

$\bullet$
For one direction,
we show the simulations of
the axioms and the inference rules of ${\sf GR}_{AX}$
in ${\sf GR}_H$

\vspace{1em}

$\blacktriangleright$ On the axioms
of ${\sf GR}_{AX}$.
The formula images of the end sequents
of the following proofs are those axioms.

\begin{center}
$\infer=[merge]{
\Box(A\supset B), \Box A
\to \Box B
}{
\infer[\Box:r]{
\Box(A\supset B)\to | \Box A \to |
\to \Box B
}{\infer[iw]{
\Box(A\supset B)\to | \Box A \to |
\Box B \Rightarrow B}
{
\infer=[K]{
\Box(A\supset B)\to | \Box A \to |
\Rightarrow B
}{\infer[\supset:l]{
A\supset B, A
\Rightarrow B}{
\infer[\Rightarrow]{
A \Rightarrow A
}{A\to A}
&
\infer[\Rightarrow]{
B \Rightarrow B
}{B\to B}
}
}
}
}
}
$
\hspace{2em}
$\infer[merge]{
\Box(\Box A\supset A)
\to \Box A
}{
\infer[\Box:r]{
\Box(\Box A\supset A)\to | \to \Box A
}{ \infer[K]{
\Box(\Box A\supset A)\to |
\Box A \Rightarrow A }
{
\infer[\supset:l]{
\Box A\supset A, \Box A
\Rightarrow A
}{
\infer[\Rightarrow]{
\Box A \Rightarrow \Box A 
}{\Box A \rightarrow \Box A }}&
\infer[\Rightarrow]{
A \Rightarrow A
}{A\to A}
}
}
}
$
\end{center}

\begin{center}
$
\infer[merge]{\blacksquare A
\to \Box A}{
\infer[\Box:r]{
\blacksquare A \to | \to \Box A}{
\infer[iw]{
\blacksquare A \to | \Box A \Rightarrow A
}{
\infer[K^{\blacksquare}]{
\blacksquare A \to | \Rightarrow A
}{
\infer[\Rightarrow]{
A \Rightarrow A
}{A\to A}
}
}
}
}
$
\hspace{2em}
$\infer[merge]{\Box A \to\Box
\blacksquare A
}{
\infer[\Box:r]{
\Box A \to | \to \Box \blacksquare A}{
\infer[iw]{
\Box A \to | \Box \blacksquare A
\Rightarrow \blacksquare A
}{
\infer[\blacksquare:r_2]{
\Box A \to | \Rightarrow \blacksquare A
}{
\infer[K]{
\Box A \to | \Rightarrow A
}{
\infer[\Rightarrow]{
A \Rightarrow A
}{A\to A}}
}
}
}
}$
\end{center}

\begin{center}
$
\infer[merge]{\Box A \to \blacksquare A,
\Box \bot}{
\infer[\Box:r]{\Box A \to | \to \blacksquare A
| \to \Box \bot}{
\infer=[iw]{\Box A \to |\to \blacksquare A
| \Box \bot\Rightarrow \bot}{
\infer[\blacksquare:r_1]{\Box A \to | \to
\blacksquare A
| \Rightarrow}{
\infer[iw]{\Box A \to | \Box A\Rightarrow A}{
\infer[K]{\Box A \to | \Rightarrow A}{
\infer[\Rightarrow]{
A \Rightarrow A
}{A\to A}
}
}
}
}
}
}
$
\hspace{2em}
$\infer[merge]{\Box\neg A \to\Box
\neg\blacksquare A
}{
\infer[\Box:r]{
\Box \neg A \to | \to \Box \neg\blacksquare A}{
\infer[iw]{
\Box \neg A \to | \Box
\neg \blacksquare A
\Rightarrow \neg \blacksquare A
}{
\infer[\neg:r]{
\Box \neg A \to | \Rightarrow \neg
\blacksquare A
}{
\infer[\blacksquare:l]{
\Box \neg A \to | \blacksquare
A\Rightarrow
}{
\infer[K]{\Box \neg A\to |
A\Rightarrow}{
\infer[\neg:l]{\neg A, A \Rightarrow}{
\infer[\Rightarrow]{
A \Rightarrow A
}{A\to A}}
}
}
}
}
}
}$
\end{center}

$\blacktriangleright$ On the inference rules
of ${\sf GR}_{AX}$.
The following are derivations of $\to \Box A$ from $\to A$
and of $\to A$ from $\to \Box A$.

\begin{center}
$
\infer[\Box:r]{
\to \Box A}{
\infer[iw]{
\Box A \Rightarrow A}{
\infer[\Rightarrow]{
\Rightarrow A
}{\to A
}
}
}
$
\hspace{2em}
$\infer[\rightarrow]{ \to A
}{
\infer[\bot]{
\Rightarrow A}{
\infer[cut]{
\to|\Rightarrow A
}{
\to \Box A
&
\infer[K]{
\Box A \to | \Rightarrow A
}{
\infer[\Rightarrow]{
A \Rightarrow A
}{A\to A}}
}}
}
$
\end{center}

$\bullet$
For the other direction, we work within
${\sf GR}_{AX}$
to
simulate the inference rules for modality
of ${\sf GR}_{H}$.
First, we treat the rules of $K$, $4$, and $K^\blacksquare$,
$4^\blacksquare$
uniformly.
These rules are expressed as follows.

\begin{center}
$
\infer
{{\cal H}| \nabla A \rightarrow | \Gamma \Rightarrow \Delta
}
{{\cal H}|\nabla^* A, \Gamma \Rightarrow \Delta
}$
\end{center}

Here (and below), $\nabla$ denote $\Box$ or $\blacksquare$,
and $\nabla^*$ denotes $\nabla$ or an empty space.

Suppose that $H\vee\Box(\nabla^* A \wedge G. \supset D)$
holds.
In terms of normality,
$H\vee. \Box\nabla^* A \supset \Box (G \supset D)$ holds.
Then, if $\nabla A \supset \Box\nabla^* A $ holds,
the desired formula is obtained:
$H\vee \neg\nabla A \vee \Box (G \supset D)$.

\vspace{1em}

$\blacktriangleright$ On the $K$-rule,
$\nabla = \Box$ and $\nabla^* $ is empty.
Then,
$\nabla A \supset \Box\nabla^* A = \Box A \supset \Box A $
surely holds.

\vspace{1em}
$\blacktriangleright$ On the $4$-rule,
$\nabla^* =\nabla = \Box$.
Then, $\nabla A \supset \Box\nabla^* A = \Box A \supset \Box\Box A $
surely holds.
(It is well known that this formula called $4$-axiom is
derivable in the sub-logic {\sf GL}.
See, for example, Boolos \cite{boolos}, pp. 11-12.)

\vspace{1em}
$\blacktriangleright$ On the $K^\blacksquare$-rule,
$\nabla = \blacksquare$ and $\nabla^* $ is empty.
Then, $\nabla A \supset \Box\nabla^* A = \blacksquare A \supset \Box A $
is surely an axiom of {\sf GR}.

\vspace{1em}
$\blacktriangleright$ On the $4^\blacksquare$-rule,
$\nabla = \nabla^* =\blacksquare$.
Then, $\nabla A \supset \Box\nabla^* A = \blacksquare A \supset \Box
\blacksquare A $.
This is derivable from the axioms of {\sf GR}:
$\blacksquare A \supset \Box A $ and $\Box A \supset \Box
\blacksquare A $.

\vspace{1em}

Next we treat the $\Box:r$-, $\blacksquare:r_1$-,
$\blacksquare:l$- and $\blacksquare:r_2$-
rules.

\vspace{1em}
$\blacktriangleright$ On the $\Box:r$-rule,
suppose that $H\vee \Box(\Box A\to A)$
holds.
Then,
we have
$H\vee \Box A$,
in terms of
the axiom
$\Box(\Box A\to A)\supset \Box A $.

\vspace{1em}
$\blacktriangleright$ On the $\blacksquare:r_1$-rule,
suppose that $H\vee \Box(\Box A\to A)$
holds.
Using the axiom
$\Box(\Box A\to A)\supset \Box A $,
we obtain
$H\vee \Box A$.
Further, using the axiom
$\Box A \supset (\Box\bot \vee
\blacksquare A)$,
we obtain $H\vee \Box\bot \vee
\blacksquare A$.

\vspace{1em}
$\blacktriangleright$ On the $\blacksquare:l$-rule,
suppose that $H\vee \Box\neg A$ holds.
The axiom
$\Diamond \blacksquare A \supset \Diamond A $
is equivalent to
$\Box\neg A \supset \Box\neg \blacksquare A $.
Thus, we obtain
$H\vee \Box\neg \blacksquare A$.

\vspace{1em}
$\blacktriangleright$ On the $\blacksquare:r_2$-rule,
suppose that $H\vee \Box A$ holds.
Using the axiom
$\Box A \supset \Box \blacksquare A $,
we obtain
$H\vee \Box\blacksquare A$.

\vspace{1em}
We omit the other cases.
\QED

\vspace{1em}
{\bf Note}
The rules of $4$ and $4^{\blacksquare}$
are permissible in {\sf GR}$_{H}$,
as they are not used in the above simulations of 
the system {\sf GR}$_{AX}$.
However, they are necessary for the system {\sf GR}$_{H}$
to enjoy the cut-elimination.
\footnote{
This fact was pointed out by
an anonymous referee.}
For example, the following hypersequent
is not cut-free provable without using 
$4$ and $4^{\blacksquare}$, respectively.

\begin{center}
 $\Box p, \blacksquare q\to |\Rightarrow \Box (p\wedge q)$   
\footnote{
In Technical Report 2 in the last part of the paper,
we show the cut-elimination procedure
for this hypersequent as an example.
}
\end{center}


\section{Cut-Elimination for {\sf GR} }

In this section,
we establish the cut-elimination theorem for ${\sf GR}_H$.
From now on,
we call our hypersequent system
{\sf GR}$_H$
just ${\sf GR}$.

\begin{theorem} \label{cut-elimination}
For any hypersequent,
it is provable in {\sf GR}
if and only if
it is provable without using the {\it cut}-rule in {\sf GR}.
\end{theorem}

Before the proof of
Theorem \ref{cut-elimination}
in \S 3.3,
we show a sort of normal form theorem 
(Lemma \ref{standard})
on 
certain inference rules on modalities, $\Box$ and $\blacksquare$,
in a proof in {\sf GR} in \S 3.1
and we process the elimination of
what is called
the diagonal formula of $\Box :r$-rule.
in \S 3.2.

\subsection{Normal Forms of Inference Rules on Modalities}

In this subsection,
we introduce a sort of normal forms of the inference rules
introducing modalities in the right hand side of a sequent,
that is,
$\Box:r$, $\blacksquare:r_1$, $\blacksquare:r_2$,
and show that applications of those rules can be transformed
into those forms in a proof.

\begin{lemma} \label{standard}
For each of the rules
$\Box:r$, $\blacksquare:r_1$,
$\blacksquare:r_2$, and
$\blacksquare:l$,
it can be assumed to appear
in the following form, say the {\rm standard} form, in a proof $P$.
\footnote{ 
When $\Psi=
(D_1, \ldots, D_k)$,
$ \nabla\Psi$ denotes $(\nabla D_1, \ldots, \nabla D_k)$,
where $\nabla $ is not necessarily uniform for formulas in $\Psi$.
E.g., it is possible that $\nabla
(D_1, D_2, D_3)
=(\Box D_1, \Box D_2, \blacksquare D_3)$.
}

\begin{center}
$\infer[\Box:r]{
\nabla\Lambda\to|
\to \Box B
}{
\nabla\Lambda\to|
\Box B \Rightarrow B}
$
\hspace{2em} $\infer[\Box:r]{
{\cal K}| 
\to \Box B
}{
\infer=[iw]{{\cal K}| 
\Box B \Rightarrow B}
{
{\cal K}| \Rightarrow
}
}
$

\vspace{1em}
$\infer[\blacksquare:r_1]{
\nabla\Lambda\to|
\to \blacksquare B |\Rightarrow
}{
\nabla\Lambda\to|
\Box B \Rightarrow B
}$

\vspace{1em}$\infer[\blacksquare:r_2]{
\nabla\Lambda\to|
\Rightarrow \blacksquare B
}{
\nabla\Lambda\to|
\Rightarrow B }
$

\vspace{1em}
$\infer[\blacksquare:l]{
\nabla\Lambda\to|
\blacksquare B \Rightarrow
}{
\nabla\Lambda\to|
B\Rightarrow }
$

\end{center}

\noindent
Here, the conditions are inflicted;

1. In each form, there is no application of {\rm ew}
introducing the indicated turnstile
of $\Rightarrow$ in the proof;

2. In the second form
(that is, the one with environment
${\cal K}$),
there are applications of $\blacksquare: r_1$ introducing
the indicated turnstile $\Rightarrow$ in the proof;

3. In the four forms except for
the second form
(that is, the ones without environment
${\cal K}$),
the indicated turnstile $\Rightarrow$ is introduced only by
the $\Rightarrow$-rule in the proof.
\end{lemma}

\noindent {\bf Proof.}
We proceed from top to bottom in $P$.
First,
handle the case of $\Box:r$.
Consider an application of
$\Box:r$, say $I$, in $P$.

\begin{center}
$
\infer[I]{
{\cal H}| \rightarrow \Box B
}{
{\cal H}| \Box B \Rightarrow^{\tau } B
}
$
\end{center}

\noindent
Here,
all the applications of $\Box:r$, $\blacksquare:r_1$,
$\blacksquare:r_2$, and
$\blacksquare:l$ above $I$
already have the standard forms;
the indicated $\Rightarrow$
of the upper sequent of $I$
is labeled with $\tau$.

Further we assume that all the ancestors of
$\Rightarrow^\tau$ of the upper sequent of $I$
below $\blacksquare:r_2$ and $\blacksquare:l$
are labeled with $\tau$.

\footnotesize

\begin{center}
\shortstack{
$\infer[\Rightarrow]
{\Gamma_0\Rightarrow^{\tau}\Delta_0}{
\Gamma_0\rightarrow\Delta_0
}$
\hspace{8em}
$\infer[ew]
{{\cal H}_2 | \Gamma_2\Rightarrow^{\tau}\Delta_2
}
{
{\cal H}_2
}
$
\hspace{8em}
$\infer[\blacksquare:l]{
{\cal H}_4 |
\blacksquare B \Rightarrow^\tau
}{{\cal H}_4 | B \Rightarrow}$
\\
\hspace{-2em}
$\vdots$
\hspace{4em}
$\infer[\blacksquare:r_1]
{{\cal H}_1 |
\to \blacksquare C | \Rightarrow^{\tau}
}
{
{\cal H}_1 |
\Box C \Rightarrow C
} $
\hspace{1.2em}
$\vdots$
\hspace{3.8em}
$\infer[\blacksquare:r_2]{
{\cal H}_3 |
\Rightarrow^\tau \blacksquare B
}{{\cal H}_3 | \Rightarrow B }$
\hspace{3em}
$\vdots$
\\
\hspace{-2em}
$\vdots$
\hspace{8em}
$\vdots$
\hspace{6em}
$\vdots$
\hspace{6em}
$\vdots$
\hspace{8em}
$\vdots$
\\
\hspace{-1.6em}
$\ddots$
\hspace{6.8em}
$\vdots
$
\hspace{6em}
$
\vdots$
\hspace{6.8em}
$
\vdots
\hspace{6.8em}
\iddots
$
\vspace{1em}
\\
\hspace{-1.6em}
$\ddots
\hspace{4.8em}
\ddots
\hspace{3.5em}
\vdots
$
\hspace{3.5em}
$
\iddots
\hspace{4.8em}\iddots
$
\vspace{2em}
\\
$\infer[I]{
{\cal H}| \rightarrow \Box B
}{
{\cal H}| \Box B \Rightarrow^{\tau } B
}$
\hspace{1em}
}
\end{center}

\normalsize
Thus, each path of $\Rightarrow^\tau$, say $\tau$-$path$, starts
with applications of $\Rightarrow$, $ew$,
$\blacksquare:r_1$,
$\blacksquare:r_2$ or
$\blacksquare:l$.
Those applications of
$\blacksquare:r_1$,
$\blacksquare:r_2$ and
$\blacksquare:l$ have the standard forms
and thus the accompanying environments ${\cal H}_i$ have
specific forms given in Definition \ref{standard}.

We say that an application of rule {\it is in} a $\tau$-path,
when the upper or lower hypersequent of it
contains $\Rightarrow^{\tau}$ of the $\tau$-path.
Thus, each of the lower and the upper sequents of
an application in a $\tau$-path may have any of
$\Rightarrow^{\tau}$, $\Rightarrow$
and $\rightarrow$ (where $\Rightarrow$ means the turnstile $\Rightarrow$
having no label $\tau$).
Then, we naturally say that
an application of rule in a $\tau$-path
is below or above
another application of rule in the
$\tau$-path.

We consider
 the following abbreviated form of $I$. 

\begin{center}
$
\infer[I]{
{\cal H}| \rightarrow \Box B
}{
{\cal H} | (\Box B)^n \Rightarrow^{\tau } B^m
}
$
\end{center}

Here, 
$X^n$ $(n \geq 0)$ denotes $n$-many occurrences of formula $X$
and the above form can be viewed as serial applications of $iw, ic$ 
and $\Box:r$.
The sequent $\Rightarrow$ is equal to $ (\Box B)^0 \Rightarrow B^0$.

Let us consider the following two groups
of applications of rule in a $\tau$-path.

\vspace{1em}
($\natural 1$) The propositional inference rules and the rules of
$iw$, $ic$
such that their lower and upper sequents
have $\Rightarrow^{\tau}$

($\natural 2$) The rules of $K$, $4$, $K^\blacksquare$, $4^\blacksquare$
such that their upper sequents
have $\Rightarrow^{\tau}$

\vspace{1em}

Note that $merge$ is excluded from 
($\natural 1$) and ($\natural 2$), 
whether its upper sequents
have $\Rightarrow^\tau$ or not.

We are going to make proof-transformations 
so that
each $\tau$-path will consist of the applications of
the following rules
in the top-down order.

\vspace{1em}
$1$. The start rule, of which the lower sequent has $\Rightarrow^\tau$;

\vspace{1em}
$2$. Rules of $(\natural 1)$, of which the upper and lower sequents have $\Rightarrow^\tau$;

\vspace{1em}
$3.$ Rules of ($\natural 2$), of which the upper sequent has $\Rightarrow^\tau$;

\vspace{1em}
$4.$ $I$, of which the upper sequent has $\Rightarrow^\tau$;

\vspace{1em}

Transformation is basically performed using the permutation technique
in two stages.
First we establish the following property.

\vspace{1em}
(*) For each $\tau$-path, between
the start rule and $I$,
there occur only applications of $(\natural 1)$ or $(\natural 2)$.

\vspace{1em}
The following is the list of
the rules that we are going to remove from the applications
between $I$ and the start rules
in $\tau$-paths
so that they are all below $I$.

\vspace{1em}
\hspace{1em}
(a)
The propositional rules, $iw, ic, ew$ and $split$
whose 
lower sequents do not have $\Rightarrow^{\tau}$


\vspace{1em}
\hspace{1em}
(b)
$K, 4, K^{\blacksquare},
4^{\blacksquare}, \Box:r, \blacksquare:r_1,
\blacksquare:r_2$ and $\blacksquare:l$
whose upper sequents do not have $\Rightarrow^{\tau}$

\vspace{1em}
\hspace{1em}
(c) $merge$

\vspace{1em}

Let $L$ be an uppermost application of a rule of
$(\natural 1)$ or $(\natural 2)$
such that there are applications of rules 
from (a, b, c)
above $L$.

Then, let $M$ denote the application just above $L$.
$M$ is one from (a, b, c).
We permute $L$ and $M$.
This is possible because, 
in the cases of (a, b),
the upper and lower sequents
of $M$ do not have $\Rightarrow^{\tau}$;
the lower sequent of $M$ is not the upper sequent of $L$.
The following is an example of the permutation 
as to (a).

\begin{center}$
\infer[L]{{\cal P}| {\cal Q}|
\Sigma_0, \Sigma_1 \gg
\Omega_0, \Omega_1,
C\wedge D | \Psi
\Rightarrow^{\tau} \Xi, \neg E
}{
\infer[M(\wedge:r)]{
{\cal P}| {\cal Q}|
\Sigma_0, \Sigma_1 \gg
\Omega_0, \Omega_1,
C\wedge D |
 E, \Psi
\Rightarrow^{\tau} \Xi
}
{ {\cal P}|
\Sigma_0 \gg \Omega_0, C
&
{\cal Q}|
\Sigma_1 \gg \Omega_1, D  |
 E, \Psi
\Rightarrow^{\tau} \Xi
}
}
$

\vspace{1em}
$\bigtriangledown$

\vspace{1em}
$
\infer[M(\wedge:r)]{
{\cal P}| {\cal Q}|
\Sigma_0, \Sigma_1 \gg
\Omega_0, \Omega_1,
C\wedge D |
\Psi
\Rightarrow^{\tau} \Xi, \neg E
}
{ {\cal P}|
\Sigma_0 \gg \Omega_0, C
&
\infer[L]{{\cal Q}|
\Sigma_1 \gg \Omega_1, D  |
 \Psi
\Rightarrow^{\tau} \Xi, \neg E}{
{\cal Q}|
\Sigma_1 \gg \Omega_1, D  |
 E, \Psi
\Rightarrow^{\tau} \Xi}
}
$
\end{center}


Also, we take an example where $L$ is from (b).

\begin{center}$
\infer[L]{{\cal P}| 
\blacksquare C \to | \Sigma_0 \Rightarrow
\Omega_0 | \Psi
\Rightarrow^{\tau} \Xi, \neg E
}{
\infer[M(K^{\blacksquare})]{
{\cal P}| 
\blacksquare C \to | \Sigma_0 \Rightarrow
\Omega_0 |
 E, \Psi
\Rightarrow^{\tau} \Xi
}
{ {\cal P}| 
C, \Sigma_0 \Rightarrow
\Omega_0 |
 E, \Psi
\Rightarrow^{\tau} \Xi
}
}
$

\vspace{1em}
$\bigtriangledown$

\vspace{1em}
$
\infer[M(K^{\blacksquare})]{{\cal P}| 
\blacksquare C \to | \Sigma_0 \Rightarrow
\Omega_0 | \Psi
\Rightarrow^{\tau} \Xi, \neg E
}{
\infer[L]{
{\cal P}| 
C, \Sigma_0 \Rightarrow
\Omega_0 |
  \Psi
\Rightarrow^{\tau} \Xi, \neg E
}
{ {\cal P}| 
C, \Sigma_0 \Rightarrow
\Omega_0 |
 E, \Psi
\Rightarrow^{\tau} \Xi
}
}
$
\end{center}


Discuss the case when $M$ is $merge$ (c).
If the lower sequent of $M$ is not 
the upper sequent of $L$,
it is similar to the other cases as follows.

\begin{center}$
\infer[L]{{\cal P}|
\Sigma_0, \Sigma_1 \gg
\Omega_0, \Omega_1
 | \Psi
\Rightarrow^{\tau} \Xi, \neg E
}{
\infer[M(merge)]{
{\cal P}| 
\Sigma_0, \Sigma_1 \gg
\Omega_0, \Omega_1|
 E, \Psi
\Rightarrow^{\tau} \Xi
}
{ {\cal P}|
\Sigma_0 \gg \Omega_0 |
\Sigma_1 \gg \Omega_1 |
 E, \Psi
\Rightarrow^{\tau} \Xi
}
}
$

\vspace{1em}
$\bigtriangledown$

\vspace{1em}
$
\infer[M(merge)]{
{\cal P}| 
\Sigma_0, \Sigma_1 \gg
\Omega_0, \Omega_1|
 E, \Psi
\Rightarrow^{\tau} \Xi, \neg E
}
{ 
\infer[L]
{{\cal P}|
\Sigma_0 \gg \Omega_0 |
\Sigma_1 \gg \Omega_1 |
  \Psi
\Rightarrow^{\tau} \Xi, \neg E
}
{{\cal P}|
\Sigma_0 \gg \Omega_0 |
\Sigma_1 \gg \Omega_1 |
  E, \Psi
\Rightarrow^{\tau} \Xi
}
}
$
\end{center}

\noindent
Here, $\gg$ denotes $\to$,  $\Rightarrow$ or $\Rightarrow^{\tau}$.

If
the lower sequent of $M$ is 
the upper sequent of $L$,
we can permute them easily 
as the following example shows.

\begin{center}$
\infer[L]{{\cal P}|
\Sigma_0, \Sigma_1 \Rightarrow^{\tau}
\Omega_0, \Omega_1, \neg E
}{
\infer[M(merge)]{
{\cal P}| E,
\Sigma_0, \Sigma_1 \Rightarrow^{\tau}
\Omega_0, \Omega_1
}
{ {\cal P}|
E, \Sigma_0 \Rightarrow^{\tau} \Omega_0 |
\Sigma_1 \Rightarrow^{\tau} \Omega_1
}
}
$

\vspace{1em}
$\bigtriangledown$

\vspace{1em}
$
\infer[M(merge)]{
{\cal P}| 
\Sigma_0, \Sigma_1 \Rightarrow^{\tau}
\Omega_0, \Omega_1, \neg E
}
{ 
\infer[L]
{{\cal P}|
\Sigma_0 \Rightarrow^{\tau} \Omega_0, \neg E |
\Sigma_1 \Rightarrow^{\tau} \Omega_1
}
{{\cal P}|
E, \Sigma_0 \Rightarrow^{\tau} \Omega_0 |
\Sigma_1 \Rightarrow^{\tau} \Omega_1
}
}
$
\end{center}


Note that
when $M$ is $merge$ and $L$ is $I$
such that 
the  lower sequent of $M$ is the upper sequent of $L$,
$I$ is duplicated for the two $\tau$-paths as follows. 

\begin{center}$
\infer[L]{{\cal P}|
\to \Box B
}{
\infer[M(merge)]{
{\cal P}|(\Box B)^a, (\Box B)^c
\Rightarrow^{\tau } B^b, B^d
}
{ {\cal P}|
(\Box B)^a
\Rightarrow^{\tau } B^b
|
(\Box B)^c
\Rightarrow^{\tau } B^d
}
}
$

\vspace{1em}
$\bigtriangledown$

\vspace{1em}
$
\infer[ic]{{\cal P}|
\to \Box B}{
\infer[M(merge)]{{\cal P}|
\to \Box B, \Box B
}{
\infer=[L]{
{\cal P}|
\to \Box B
|
\to \Box B
}
{ {\cal P}|
(\Box B)^a
\Rightarrow^{\tau' } B^b
|
(\Box B)^c
\Rightarrow^{\tau } B^d
}
}
}
$
\end{center}

In this case, we rename  one of the two  $\tau$-paths
to $\tau'$.
Thus, the property (*) has been established.

Next we establish the following property.

\vspace{1em}
(**)
Each $\tau$-path has a $mid$-$hypersequent$, $MID$, that is,
there occur only applications of $(\natural 1)$
above $MID$ 
and there occur only applications of $(\natural 2)$
below $MID$.


\vspace{1em}
At this stage, for each $\tau$-path,
the applications of $(\natural 1)$ and $(\natural 2)$
are all mixed up
between the start rules and $I$.
Now we are going to sort out them 
so that 
the applications of $(\natural 2)$
are located only in the last part of each $\tau$-path.

The proof-transformation is similar to that for the property $(*)$.
Let $L$ be an uppermost application of a rule of
$(\natural 1)$
such that there are applications of rules of 
$(\natural 2)$ above $L$.
And let $M$ be the application just above $L$.
$M$ is one of $(\natural 2)$.
Permute $L$ and $M$.
This is possible because
the lower sequent
of $M$ 
is a $\to$-sequent and is not the upper sequent of $L$
and because the presence of the auxiliary formula,  say $C$,
of $M$ does not affect $L$
(as $L$ is propositional rule, $iw$ or $ic$).
The following is an example of the permutation.

\begin{center}$
\infer[L]{{\cal P}| {\cal Q}|
\blacksquare C\to | 
\Sigma_0, \Sigma_1 \Rightarrow^{\tau}
\Omega_0, \Omega_1,
D \wedge  E
}{
\infer[M(K^\blacksquare)]{
{\cal P}|
\blacksquare C\to | \Sigma_0 \Rightarrow^{\tau}
\Omega_0, D
}
{ {\cal P}|
C, \Sigma_0  \Rightarrow^{\tau}
\Omega_0, D
}
&
{\cal Q}|
\Sigma_1 \Rightarrow^{\tau} \Omega_1, E
}
$

\vspace{1em}
$\bigtriangledown$

\vspace{1em}
$
\infer[M(K^\blacksquare)]{{\cal P}| {\cal Q}|
\blacksquare C\to | 
\Sigma_0, \Sigma_1 \Rightarrow^{\tau}
\Omega_0, \Omega_1,
D \wedge  E
}{
\infer[L]{{\cal P}| {\cal Q}|
C,\Sigma_0, \Sigma_1 \Rightarrow^{\tau}
\Omega_0, \Omega_1,
D \wedge  E
}{
{\cal P}|
C, \Sigma_0  \Rightarrow^{\tau}
\Omega_0, D
&
{\cal Q}|
\Sigma_1 \Rightarrow^{\tau} \Omega_1, E
}
}
$
\end{center}


Thus we established the properties,  (*) and (**),
so that 
each $\tau$-path has a $MID$,
of the form:

\begin{center} 
$
{\cal H}_0| \nabla \Lambda \to | 
\nabla^*\Psi, (\Box B)^n \Rightarrow^{\tau } B^m$,
\end{center}

\noindent
where

\vspace{1em}
$\bullet$
Formulas of $\nabla \Lambda$ come from the
lower hypersequents of possible
applications of 
$\blacksquare:r_1$,
$\blacksquare:r_2$ and
$\blacksquare:l$ with which the $\tau$-paths can start.

\vspace{1em}
$\bullet$
$\nabla^* \Psi
$ is the  multiset 
of the auxiliary formulas of
possible applications of $K, 4$, $K^{\blacksquare}$, and $4^{\blacksquare}$
below the $I$.
\footnote{ 
As in the case of $\nabla\Psi$, 
$\nabla^*$ of  $\nabla^*\Psi$ is not necessarily uniform for formulas in $\Psi$.
E.g., it is possible that $\nabla^*
(D_1, D_2, D_3)
=(\Box D_1, D_2, \blacksquare D_3)$.}

\vspace{1em}
Now $P$ is as follows.


\begin{center}
\shortstack{\vspace{1em}
\hspace{-6.4em}$P_0$ \hspace{4em}\deduc \hspace{2em} \deduc \\
${\cal H}_0| \nabla \Lambda \to | 
\nabla\Psi^*, (\Box B)^n \Rightarrow^{\tau } B^m$\vspace{1em} \\
\vspace{1em}
\vdots
\\
\hspace{1em}
$\infer[I]{
{\cal H}_0| \nabla \Lambda \to | 
\nabla\Psi \to |
\rightarrow \Box B
}{
{\cal H}_0| \nabla \Lambda \to | 
\nabla\Psi \to |(\Box B)^n \Rightarrow^{\tau } B^m
}
$
}
\end{center}

\normalsize

\noindent
Note that 
those applications above $MID$s
in $\tau$-paths satisfy the property that their lower sequent 
has $\Rightarrow^{\tau}$
and those below $MID$s do not satisfy the property.

\vspace{1em}
Then, it is natural to consider a more generalized form of $I$ as 
follows.

\begin{center}$
\infer[I]{
{\cal H}_0|  \nabla \Lambda \to | 
\nabla\Psi \to | \to \Box B
}{{\cal H}_0| \nabla \Lambda \to | 
\nabla\Psi^*, (\Box B)^n \Rightarrow^{\tau } B^m}
$
\end{center}

In this form, the application $I$ directly follows the $MID$
and so there are only applications of $(\natural 1)$ in $\tau$-paths.
As $(\natural 1)$ does not include $merge$,
two $\tau$-paths can be contracted only by $\wedge:r$
whose upper sequents have $\Rightarrow^\tau$.

We prove the following statements.

\vspace{1em}
$\bullet$ When there are applications of $ew$ introducing $\Rightarrow^{\tau }$ above the $I$,
${\cal H}_0 $ is provable.

$\bullet$ When there are applications of $\blacksquare:r_1$ introducing $\Rightarrow^{\tau }$ above the $I$,
${\cal H}_0| \Rightarrow^{\tau }$ is provable.

$\bullet$ Otherwise, the form:
$\nabla \Lambda \to |
\nabla^* \Psi, (\Box B)^n
\Rightarrow^{\tau } B^m$ is provable.

\vspace{1em}
The proof is carried out 
by induction on the number of applications of the
rules in the $\tau$-paths.
As the base case,
the $\tau$-paths consist of only one of the following: 
$ew$, $\blacksquare:r_1$,
$\Rightarrow$ and $\blacksquare:r_2$. 
In the case of $ew$,
clearly ${\cal H}_0 $ is provable and thus we can delete the application of $I$.

\begin{center}
$
\infer[I\hspace{2em}\rhd \hspace{2em}]{{\cal H}_0 | \to \Box B}{
\infer[ew]{
{\cal H}_0 | (\Box B)^n \Rightarrow^{\tau } B^m}
{
{\cal H}_0
}}
$
$
\infer[ew]{
{\cal H}_0 | \rightarrow \Box B}
{
{\cal H}_0
}
$
\end{center}

In the other cases of the base case,
$I$ is already of the standard form,
as illustrated below.

\small
\begin{center}
$
\infer[I]{ \to \Box B}{
\infer[\Rightarrow]{
(\Box B)^n \Rightarrow^{\tau } B^m}
{ (\Box B)^n \rightarrow B^m
}}$
\hspace{2em}
$
\infer[I]{\nabla \Lambda \to | \to \Box \blacksquare C}{
\infer[\blacksquare:r_2]{
\nabla \Lambda \to | 
\Rightarrow ^{\tau } \blacksquare C
}{
\nabla \Lambda \to |\Rightarrow C }
}
$
\hspace{2em}
$
\infer[I]{\nabla \Lambda \to | \to \blacksquare C |\to \Box B}{
\infer[\blacksquare:r_1]{
\nabla \Lambda \to| \to \blacksquare C | \Rightarrow ^{\tau }
}{
\nabla \Lambda \to | \Box C \Rightarrow C }
}
$
\end{center}

\normalsize
Here, by the assumption, the applications of $\blacksquare:r_2$ and $\blacksquare:r_1$ above $I$
are of the standard form.
Hence, $I$ is also of the standard form.

In the induction step,
we handle only the case for $\wedge:r$.

\begin{center}$
\infer[\wedge:r]{
{\cal H}_0|{\cal H}_1|\nabla \Lambda_0 \to | \nabla \Lambda_1 \to |
\Sigma_0, \Sigma_1
\Rightarrow^{\tau } \Xi_0, \Xi_1,C\wedge D
}
{ {\cal H}_0|\nabla \Lambda_0 \to |
\Sigma_0
\Rightarrow^{\tau } \Xi_0, C &
{\cal H}_1|\nabla \Lambda_1 \to |
\Sigma_1
\Rightarrow^{\tau } \Xi_1, D}$

\end{center}

\noindent {\it Case 1}. If
there are applications of $ew$
introducing $\Rightarrow^{\tau }$
above $ {\cal H}_0| {\cal H}_1|\nabla \Lambda_0 \to | \nabla \Lambda_1 \to |
\Sigma_0, \Sigma_1
\Rightarrow^{\tau } \Xi_0, \Xi_1, C \wedge D$,
then they occur
above ${\cal H}_0|\nabla \Lambda_0 \to |
\Sigma_0
\Rightarrow^{\tau } \Xi_0, C$
or
$
{\cal H}_1| \nabla \Lambda_1 \to |
\Sigma_1
\Rightarrow^{\tau } \Xi_1, D$.
By the induction hypothesis,
$
{\cal H}_0 
$ or $
{\cal H}_1 
$ is provable.
In either case,
${\cal H}_0|{\cal H}_1 
$ is provable by using $ew$.

\noindent {\it Case 2}. If there are applications of $\blacksquare:r_1$
that introduce $\Rightarrow^{\tau }$
above $ {\cal H}_0|{\cal H}_1|\nabla \Lambda_0 \to | \nabla \Lambda_1 \to |
\Sigma_0, \Sigma_1
\Rightarrow^{\tau } \Xi_0, \Xi_1, C \wedge D$,
then they occur
above ${\cal H}_0|\nabla \Lambda_0 \to |
\Sigma_0
\Rightarrow^{\tau } \Xi_0, C$
or
$
{\cal H}_1| \nabla \Lambda_1 \to |
\Sigma_1
\Rightarrow^{\tau } \Xi_1, D$.
Therefore, by the induction hypothesis,
$
{\cal H}_0|
\Rightarrow^{\tau } 
$ or $
{\cal H}_1|
\Rightarrow^{\tau } 
$ is provable.
In either case,
${\cal H}_0|
{\cal H}_1|
\Rightarrow^{\tau } $ is provable by using $ew$.

\noindent {\it Case 3}. Otherwise,
there is no application of $ew$ or $\blacksquare:r_1$
introducing $\Rightarrow^{\tau }$
above ${\cal H}_0|\nabla \Lambda_0 \to |
\Sigma_0
\Rightarrow^{\tau } \Xi_0, C$
or
$
{\cal H}_1| \nabla \Lambda_1 \to |
\Sigma_1
\Rightarrow^{\tau } \Xi_1, D$.
Then, by the induction hypothesis,
$\nabla \Lambda_0 \to |
\Sigma_0 \Rightarrow^{\tau } \Xi_0, C$
and $ \nabla \Lambda_1 \to |
\Sigma_1 \Rightarrow^{\tau } \Xi_1, D$
are provable,
and thus
$ \nabla \Lambda_0 \to | \nabla \Lambda_1 \to |
\Sigma_0, \Sigma_1 \Rightarrow^{\tau } \Xi_0, \Xi_1, C \wedge D$ is also provable
by using $\wedge:r$.

\vspace{1em}
In this way,
the application $I$ of $\Box:r$
is deleted or made of the standard form.
For the other cases of
$\blacksquare:r_1$,
$\blacksquare:r_1$ and
$\blacksquare:l$,
we can make similar proof-transformations
to make them of the standard form.
For example, for the case of $\blacksquare:l$,
we first use the permutation technique so that
each $MID$ is directly followed by rules of $K, 4, K^{\blacksquare}, 4^{\blacksquare}$
and the application $I$ of $\blacksquare:l$ under consideration.
Then, consider the following abbreviated form for $I$.

\begin{center}
$
\infer[I]{
{\cal H}_0| \nabla \Lambda \to | \nabla\Psi\to |\blacksquare B \Rightarrow
}{
{\cal H}_0| \nabla \Lambda \to | \nabla\Psi^*, B^n \Rightarrow^{\tau }
}
$
\end{center}

Then, we can prove the statements corresponding to the above ones
in a similar way to show that $I$ can be deleted or made of the 
standard forms.

\vspace{1em}
Here we put some remark for the
cases of $\blacksquare:r_1, \blacksquare:r_2$
and $\blacksquare:l$.
For each of these cases,
as in the case of $\Box:r$,
we obtain the two forms as follows.

\begin{center}
$\infer[\blacksquare:r_1]{
\nabla \Lambda\to |
\to \blacksquare B |\Rightarrow
}{
\nabla \Lambda\to |
\Box B \Rightarrow B
}
$
\hspace{2em}
$\infer[\blacksquare:r_1]{
{\cal K}| \to \blacksquare B |\Rightarrow
}{
{
\infer=[iw]{
{\cal K}|\Box B \Rightarrow B}
{{\cal K}| \Rightarrow}
}
}
$

\vspace{1em}$\infer[\blacksquare:r_2]{
\nabla\Lambda\to|
\Rightarrow \blacksquare B
}{
\nabla\Lambda\to|
\Rightarrow B }
$
\hspace{2em}
$\infer[\blacksquare:r_2]{
{\cal K}| \Rightarrow \blacksquare B
}{
{
\infer=[iw]{
{\cal K}| \Rightarrow B}
{{\cal K}| \Rightarrow}
}
}
$

\vspace{1em}
$\infer[\blacksquare:l]{
\nabla\Lambda\to|
\blacksquare B \Rightarrow
}{
\nabla\Lambda\to|
B\Rightarrow }
$\hspace{2em}
$\infer[\blacksquare:l]{
{\cal K}| \blacksquare B \Rightarrow
}{
{
\infer=[iw]{
{\cal K}|B \Rightarrow }
{{\cal K}| \Rightarrow}
}
}
$
\end{center}

However, 
we can reduce the above three forms with the environment $\cal K$
to applications of $ew, iw$ in the following way.

\begin{center}
$\infer[\blacksquare:r_1\hspace{2em}\rhd \hspace{2em}]{
{\cal K}| \to \blacksquare B |\Rightarrow
}{
{
\infer=[iw]{
{\cal K}|\Box B \Rightarrow B}
{{\cal K}| \Rightarrow}
}
}$
$
\infer[ew]{
{\cal K}|\to \blacksquare B |\Rightarrow}
{
{\cal K}| \Rightarrow
}
$
\end{center}

\begin{center}
$\infer[\blacksquare:r_2\hspace{2em}\rhd \hspace{2em}]{
{\cal K}| \Rightarrow \blacksquare B
}{
{
\infer[iw]{
{\cal K}| \Rightarrow B}
{{\cal K}| \Rightarrow}
}
}
$
$
\infer[iw]{
{\cal K}|\Rightarrow \blacksquare B }
{
{\cal K}| \Rightarrow
}
$
\end{center}

\begin{center}
$\infer[\blacksquare:l\hspace{2em}\rhd \hspace{2em}]{
{\cal K}| \blacksquare B \Rightarrow
}{
{
\infer=[iw]{
{\cal K}|B \Rightarrow }
{{\cal K}| \Rightarrow}
}
}
$
$
\infer[iw]{
{\cal K}| \blacksquare B \Rightarrow }
{
{\cal K}| \Rightarrow
}
$
\end{center}

Thus, we do not need to consider the standard forms with the environment ${\cal K}$
for the cases of $\blacksquare:r_1, \blacksquare:r_2$
and $\blacksquare:l$.
\QED

\subsection{Elimination of the diagonal formula}

Now we set about the
elimination of the diagonal formula
of an application of $\Box:r$
or $\blacksquare:r_1$ in a proof $P$ in {\sf GR},
which is required for
the cut-elimination for {\sf GR}.
This situation is similar to the cut elimination case for {\sf GL},
and we apply the method developed
for extensions of {\sf GL}
in \cite{kushida2020, kashima}.

Let $I_0$ be an application
of $\Box:r$ or $\blacksquare:r_1$
in a proof $P$.
In terms of Lemma 
\ref{standard},
$I_0$ has one of the three standard forms.
There is no difficulty for the elimination
for the standard form of $\Box:r$ with the environment ${\cal K}$.
The other two forms of $I_0$,
for which we are going to
solve the deletion of the diagonal formula,
are as follows.

\begin{center}
$\infer[\Box:r]{ \nabla \Lambda \to |
\to \Box B
}{
\nabla \Lambda \to |
\Box B \Rightarrow^{\tau_1} B}
$
\hspace{2em}
$\infer[\blacksquare:r_1]{ \nabla \Lambda \to |
\to \blacksquare B |\Rightarrow
}{
\nabla \Lambda \to |
\Box B \Rightarrow^{\tau_1} B}
$
\end{center}

Here, 
the ${\tau_1}$-paths ending with 
these $\Box:r$ or $\blacksquare:r_1$
can start with the $\Rightarrow$-, $\blacksquare:l$- and 
$\blacksquare:r_2$-rules;
the turnstile $\Rightarrow^{\tau_1}$ is introduced only by 
$\Rightarrow$-rule.

Also, we may suppose that
there is no application of either $iw$ or $ew$
introducing
the diagonal formula $\Box B$ of $I_0$
at this start stage,
and 
there are
$n$-many applications, say
$(L_1, L_2, \ldots, L_n)$, of 
the $K$-rule 
which introduce the diagonal formula $\Box B$ of $I_0$.

\small
\begin{center}
\shortstack{
$
\infer[L_i]
{{\cal K}_i|
\Box B \to
| \Sigma_i \Rightarrow \Psi_i
}
{{\cal K}_i|
B, \Sigma_i \Rightarrow \Psi_i
}$ \hspace{2em}
$\infer[L_j]
{{\cal K}_j|
\Box B \to
| \Sigma_j \Rightarrow \Psi_j
}
{{\cal K}_j|
B, \Sigma_j \Rightarrow \Psi_j
}$
\\
$\vdots$
\hspace{11em} $\vdots$
\\
$\vdots$
\hspace{11em} $\vdots$\\
$\vdots$\hspace{11.3em} $\vdots$
\\
\hspace{.8em}
$\ddots$\hspace{3em}
\hspace{4em}
$\iddots$ \hspace{1em}
\vspace{1em}
\\
\hspace{-1em}
$
\nabla \Lambda \to |
\Box B \Rightarrow^{\tau_1} B
$}
\end{center}

\normalsize


The turnstile $\rightarrow$ of
the lower sequent of an $L_i$ 
must be converted to $\Rightarrow$ by the $\Rightarrow$-rule.
Let $N_i$ be the lowest one of such applications of
the $\Rightarrow$-rule.
The upper hypersequent of
$N_i$ must be a $\to$-sequent
and thus the turnstile $\Rightarrow$ of
$\Sigma_i \Rightarrow \Psi_i$
(in the lower hypersequent of $L_i$)
must be converted to $\to$ by
an application, say $M_i$, of
the $\Box:r$-
or $\blacksquare:r_1$-rule
above $N_i$.

\small

\begin{center}
\shortstack{
$
\infer[L_i]
{{\cal K}_i|
\Box B \to
| \Sigma_i \Rightarrow \Psi_i
}
{{\cal K}_i|
B, \Sigma_i \Rightarrow \Psi_i
}$ \hspace{4em}
$\infer[L_j]
{{\cal K}_j|
\Box B \to
| \Sigma_j \Rightarrow \Psi_j
}
{{\cal K}_j|
B, \Sigma_j \Rightarrow \Psi_j
}$\\
\hspace{-5.7em}
$Q_i[\Box B]$ \hspace{2em}$\vdots$
\hspace{9em}
$Q_j[\Box B]$ \hspace{2em}$\vdots$
\\
$\vdots$
\hspace{14.5em}$\vdots$
\vspace{1em}
\\
\hspace{1em}$
\infer[M_i]{{\cal J}_i| \Box B, X_i \rightarrow \Xi_i | \to \blacksquare C_i,
| \Rightarrow
}{
{\cal J}_i| \Box B, X_i \rightarrow \Xi_i| \Box C_i \Rightarrow C_i
}
$\hspace{2em}
$
\infer[M_j]{{\cal J}_j| \Box B,
X_j\rightarrow \Xi_j| \to \Box C_j}{
{\cal J}_j| \Box B, X_j \rightarrow \Xi_j| \Box C_j \Rightarrow C_j
}
$
\\$\vdots$
\hspace{14em} $\vdots$
\\$\vdots$\hspace{14.2em} $\vdots$ \vspace{1em}
\\
\hspace{1em}$
\infer[N_i]{\Box B, \Phi_i \Rightarrow^{\tau_1} \Omega_i}{
\Box B, \Phi_i \rightarrow \Omega_i
}
$
\hspace{7em}
$
\infer[N_j]{\Box B, \Phi_j\Rightarrow^{\tau_1} \Omega_j}{
\Box B, \Phi_j \rightarrow \Omega_j
}
$
\\$\vdots$ \hspace{14em} $\vdots$
\\
\hspace{-5.3em}
\hspace{4.77em}
$\vdots$ \hspace{14em}
$\vdots$\vspace{1em}
\\
\hspace{.8em}
$\ddots$\hspace{3em}
\hspace{4em}
$\iddots$ \hspace{1em}
\vspace{1em}
\\
$
\infer[I_0]{ \nabla \Lambda \to | 
\rightarrow \Box B}{
\nabla \Lambda \to |\Box B 
\Rightarrow^{\tau_1} B}
$ \vspace{1em}
\\ \vspace{1em}
}
\end{center}

\normalsize
Note that $M_i$ and $M_j$ cannot be the $\to$-rule
because of the presence of the $\to$-sequent in their lower hypersequent.

Also note that
due to the presence of $\Box B$,
there is no application of the $\blacksquare:l$- or $\blacksquare:r_2$-rule
below each $N_i$.

\vspace{1em}
For each $1\leq i \leq n$,
let $Q_i[\Box B]$ be
the subproof ending with $M_i$;
'$[\Box B]$' of $Q_i[\Box B]$
denotes the occurrence of the diagonal
formula $\Box B$
in the lower hypersequent of $M_i$.
$Q_i[\Gamma]$ denotes
a proof of the same hypersequent
as $Q_i[\Box B]$
except that $\Box B$ is replaced by
a 
multiset
 of formulas $\Gamma$.

Now, we set $R[\Box B]$ to be
the proof-fragment
starting with the lower hypersequent of
each $M_i$ and
ending with
$\nabla \Lambda\to | \Box B 
\Rightarrow^{\tau_1} B$. 
'$[\Box B]$' of $R[\Box B]$ indicates
the occurrence of the diagonal
formula $\Box B$
in the sequent
$\Box B 
\Rightarrow^{\tau_1} B$. 
$R[\Gamma]$ denotes
a proof-fragment ending with
the same sequent as $R[\Box B]$
except that $\Box B$ is replaced by
a 
multiset
 of formulas $\Gamma$.


\begin{lemma}\label{basic}
For any $1\leq i \leq n$,
the subproof $Q_i[ 
\Box C_i]$ is obtained.
\end{lemma}

\noindent {\bf Proof.}
We can construct proofs of the hypersequent obtained from
the lower hypersequent of $M_i$ by replacing $\Box B$ with $
\Box C_i
$
in the following way.

\vspace{1em}
$\bullet$ When $M_i$ is
$\blacksquare:r_1$:

\small
\begin{center}
$
\infer[M_i\hspace{1.5em}\rhd
\hspace{1.5em}]{{\cal J}_i| \Box B,
X_i
\rightarrow \Xi_i | \to \blacksquare
C_i | \Rightarrow }{
{\cal J}_i| \Box B, X_i \rightarrow \Xi_i| \Box C_i \Rightarrow C_i
}
\infer=[iw,ew]{{\cal J}_i| 
\Box C_i,
X_i
\rightarrow \Xi_i | \to \blacksquare
C_i | \Rightarrow }{
\infer[\blacksquare:r_1]{
\Box C_i \to |
\to \blacksquare
C_i | \Rightarrow }{\infer[iw]{
\Box C_i \to | \Box C_i \Rightarrow C_i}{
\infer[K]{\Box C_i\to | \Rightarrow
C_i}{C_i \Rightarrow C_i}
}
}}
$
\end{center}

\vspace{1em}
\normalsize
$\bullet$ When $M_j$ is
$\Box:r_1$:

\small
\begin{center}
$
\infer[M_j\hspace{1.5em}\rhd
\hspace{1.5em}]{{\cal J}_j| \Box B,
X_j\rightarrow \Xi_j| \to \Box C_j}{
{\cal J}_j| \Box B, X_j \rightarrow \Xi_j| \Box C_j \Rightarrow C_j
}
\infer=[iw,ew]{{\cal J}_i| 
\Box C_i, X_i
\rightarrow \Xi_i | \to \Box
C_i}{
\infer[\Box:r_1]{
\Box C_i \to |
\to \Box
C_i}{\infer[iw]{
\Box C_i \to | \Box C_i \Rightarrow C_i}{
\infer[K]{\Box C_i\to | \Rightarrow
C_i}{C_i \Rightarrow C_i}
}
}}
$
\end{center}

\normalsize
Thus we obtain the proofs
$Q_i[
\Box C_i]$
($1\leq i \leq n$).
\QED

\begin{lemma}
$P$ can be so transformed
that the fragment $R[
\overrightarrow{(\Box C_i)}_{1\leq i \leq n}]$
is contained.
\end{lemma}

\noindent {\bf Proof.}
By the latest lemma,
for each $1\leq i \leq n$,
the proof $Q_i[
\Box C_i]$
is obtained.
Further we add formulas by $iw$
to obtain
$Q_i[
\overrightarrow{(\Box C_p)}_{1\leq p \leq n}]$
($1\leq i \leq n$).

Then,
using the subproofs $Q_i[
\overrightarrow{(\Box C_p)}_{1\leq p \leq n}]$,
we simulate the proof-fragment
$R[\Box B]$, where the  multiset  $\overrightarrow{(\Box C_p)}_{1\leq p \leq n}$ is added
to the related sequents
so that it ends with the following.

\begin{center}
$
\nabla \Lambda \to | 
\overrightarrow{(\Box C_p)}_{1\leq p \leq n}
\Rightarrow B
$
\end{center}

Thus, the desired proof-fragment
$R[
\overrightarrow{(\Box C_p)}_{1\leq p \leq n}]$ is obtained.
\QED

\vspace{1em}
Then our main task is to reduce
$\overrightarrow{(\Box C_i)}_{1\leq i \leq n}$
of
$R[
\overrightarrow{(\Box C_i)}_{1\leq i \leq n}]$
one by one
and thus to obtain
$R[
\lambda]$.

For each $1\leq i \leq n$,
let $\alpha_i$ denote
$i$ or $\lambda$ (an empty space).
We set $\Box C_{\lambda}=\lambda$.
Further, for simplicity,
we write
$Q_i[\alpha_1, \alpha_2,
\ldots, \alpha_n]$
to signify
the proof
$Q_i[
\overrightarrow{(\Box C_{\alpha_p})}_{1\leq p \leq n}]$.

Also,
we write
$R[\alpha_1, \alpha_2,
\ldots, \alpha_n]$
to signify
the proof-fragment
$R[
\overrightarrow{(\Box C_{\alpha_p})}_{1\leq p\leq n}]$.

At this moment,
$R[1,2, \ldots, n]$ is available
by the latest lemma.
We are going to reduce $n$
-many numbers in it one by one
to obtain $R[\lambda]$.

\begin{lemma} \label{crucial}
Let $1\leq i \leq n$ be arbitrary
and $(\alpha_1, 
\ldots, \alpha_n)$
be arbitrary sequence
of this form.
Suppose that $P$ is so transformed
that
$R[\alpha_1, \ldots,
\alpha_{i-1}, i,
\alpha_{i+1},
\ldots, \alpha_n]$ is contained.
Then,
$Q_i[ \alpha_1, \ldots,
\alpha_{i-1}, \lambda,
\alpha_{i+1},
\ldots, \alpha_n]$ is obtained.
\end{lemma}

\noindent {\bf Proof.}
Using $R[\alpha_1, \ldots,
\alpha_{i-1}, i,
\alpha_{i+1},
\ldots, \alpha_n]$,
we construct the following proof.

\small

\begin{center}
\shortstack{
\hspace{6em}
\infer=[
4, 
merge
]
{
\nabla \Lambda \to |{\cal K}_i|
\overrightarrow{(\Box
C_{\alpha_p})}_{i\not =p}
\to
| \Box C_i, \Sigma_i \Rightarrow
\Psi_i
}{
\infer[cut]{
\nabla \Lambda \to |{\cal K}_i|
\Box C_i,
\overrightarrow{(\Box
C_{\alpha_p})}_{i\not = p
},
\Sigma_i \Rightarrow \Psi_i
}{
\shortstack{
\hspace{-6em}
$R[\alpha_1, \ldots,
\alpha_{i-1}, i,
\alpha_{i+1},
\ldots, \alpha_n
]$
\hspace{.5em}
\deduc
\vspace{1em}\\
$
\nabla \Lambda \to |
\Box C_i,
\overrightarrow{(\Box
C_{\alpha_p})}_{i\not = p }
\Rightarrow B$}
\hspace{1.5em}
\shortstack{
\deduc
\vspace{1em}
\\
${\cal K}_i|
B,
\Sigma_i \Rightarrow \Psi_i
$
}
}}
\\
$\vdots$
\vspace{1em}
\\
\hspace{-15.6em}
$Q_i[\alpha_1, \ldots,
\alpha_{i-1}, \lambda,
\alpha_{i+1},
\ldots, \alpha_n]
$
\hspace{2em}$\vdots$
\vspace{1em}
\\
$\vdots$
\vspace{1em}
\\
\vspace{1em}
$
\infer[M_i]{
\nabla \Lambda \to |
{\cal J}_i| 
\overrightarrow{(\Box
C_{\alpha_p})}_{i\not = p},
X_i \rightarrow \Xi_i | \to
\blacksquare C_i | \Rightarrow }{
\infer[ic]{
\nabla \Lambda \to |{\cal J}_i| 
\overrightarrow{(\Box
C_{\alpha_p})}_{i \not = p},
X_i \rightarrow \Xi_i| \Box C_i
\Rightarrow C_i
}{
\nabla \Lambda \to |{\cal J}_i| 
\overrightarrow{(\Box
C_{\alpha_p})}_{i \not = p},
X_i \rightarrow \Xi_i| \Box C_i, \Box C_i
\Rightarrow C_i}
}
$
}
\end{center}

\normalsize
Here, below serial applications of
$
4
$
and $merge$
,
the subproof
$Q_i$ is simulated
where $\Box B$
is replaced with
$
\overrightarrow{(\Box
C_{\alpha_p})}_{i \not = p}
$,
and $ \Box C_i$ is added to the sequent $\Sigma_i \Rightarrow
\Psi_i$
and $\nabla \Lambda \to |$ is added.
The proof obtained is 
surely

\begin{center}
$Q_i[\alpha_1, \ldots,
\alpha_{i-1}, \lambda,
\alpha_{i+1},
\ldots, \alpha_n]
$.
\end{center}

\normalsize
Note that instead several applications of $cut$ with $B$ the cut formula
can be added  in the obtained proof.
\QED

\begin{lemma}
Let $1\leq q \leq n$ be arbitrary 
and let $(h_1, h_2,
\ldots, h_q)$
be a sequence of arbitrary $q$-many numbers from
$\{ 1, 2,
\ldots, n\}$.
Suppose that $P$ contains
$R[h_1, h_2,
\ldots, h_q]$.
Then, for any $ h_r $
of $(h_1, h_2,
\ldots, h_q)$,
$P$ is transformed
so that
$R[ h_1, h_2,
\ldots, h_{r-1}, \lambda, h_{r+1}, \ldots, h_q]$ is contained.
\end{lemma}

\noindent {\bf Proof.}
It suffices to show that
for each $1\leq i\leq n$,
the subproof

\begin{center}
$Q_i[ h_1, h_2,
\ldots, h_{r-1}, \lambda, h_{r+1}, \ldots, h_q]$
\end{center}

\noindent
is obtained.
Because: from such subproofs $Q_i$,
we can simulate the proof-fragment $R$,
so
we obtain the desired proof.

Take any $Q_i$.
If $i$ of $Q_i$ is equal to one of
$( h_1, h_2,
\ldots, h_{r-1}, \lambda, h_{r+1}, \ldots, h_q)$,
then, by Lemma \ref{basic},
the subproof
$Q_i[h_1, h_2,
\ldots, h_{r-1}, \lambda, h_{r+1}, \ldots, h_q]$
is obtained.
So we suppose that
$i$ of $Q_i$
is 
different from any of

\begin{center}
$(h_1, h_2,
\ldots, h_{r-1}, \lambda, h_{r+1}, \ldots, h_q)$.
\end{center}

Those subproofs
are obtained
by using the following
proof-fragments 
in terms of Lemma \ref{crucial}.

\begin{center}
$R[ h_1, h_2,
\ldots, h_{r-1}, 1, h_{r+1}, \ldots, h_q]$

$R[ h_1, h_2,
\ldots, h_{r-1}, 2, h_{r+1}, \ldots, h_q]$

$\vdots$

$R[ h_1, h_2,
\ldots, h_{r-1}, h_r, h_{r+1}, \ldots, h_q]$

$\vdots$

$R[ h_1, h_2,
\ldots, h_{r-1}, n, h_{r+1}, \ldots, h_q]$

\end{center}

These fragments are, however,
supposed to be obtained by the hypothesis of
the current lemma.
\QED

\begin{corollary} \label{crucial_corollary}
$P$ can be so transformed
that $R[\lambda]$ is contained.
\end{corollary}

Again, we should note that
the number of applications of $cut$ with $B$ the cut formula
can increase after the deletion of the diagonal formula.

\subsection{Top-down cut-elimination for {\sf GR}}

Now we start the proof of Theorem \ref{cut-elimination}.
Consider a proof $P$ such that the $cut$-rule is applied only in the last part.

\begin{center}
$
\infer{{\cal H} | {\cal I} |
\Gamma, \Pi \gg \Delta, \Theta
}{
\shortstack{\deduc \vspace{.6em}\\ $
\infer[I]{{\cal H} |\Gamma \gg \Delta, A}{}
$} &
\shortstack{\deduc \vspace{.6em}\\ $
\infer[J]{
{\cal I} |A, \Pi \gg \Theta}{}$
}
}
$
\end{center}

The turnstile of $P$ is defined to be the one of
the lower sequent of the $cut$.

We may assume that there is no application of $ew$ nor $iw$
introducing the cut formula $A$.

First, we reduce the forms of initial sequent $A\Rightarrow A$ in $P$ to:

\begin{center}
$p \rightarrow p$ \hspace{2em}
and \hspace{2em}
$\blacksquare C \to \blacksquare C$.
\end{center}

As is easily checked, the other forms $A\gg A$
can be derived from these two forms.
\footnote{
The fact that $\blacksquare C \to \blacksquare C$
is underivable from $ C \to C$ in {\sf GR}
is consistent with 
the formal definition
of Rosser provability predicate.
It is not describable only by the propositional logic and the
G\"odelean provability.
}

We can also reduce the cut formula $A$ to an atomic formula or a modal formula $\Box B$ or $\blacksquare B$.
The reduction is a method due to Buss \cite{buss1998} and valid to the formulas of the
propositional and predicate logic.
We demonstrate the method for the case when $A=B\supset C$.
Given proofs of ${\cal H} |\Gamma \gg \Delta, B\supset C $
and
${\cal I} |B\supset C, \Pi \gg \Theta$,
it is not so hard to construct
proofs of the following three hypersequents.

\begin{center}
${\cal H} | B, \Gamma \gg \Delta, C $

${\cal I} | \Pi \gg \Theta, B$

${\cal I} | C, \Pi \gg \Theta$.
\end{center}

Then, we can derive the original end hypersequent
by applying $cut$ with $B, C$ the cut formulas.
In this way, we can assume that the cut formula $A$ is not of the propositional form.
However, this method is not applicable to the case where the cut formula is a modal formula
(cf. Kushida \cite{kushida2024,
kushida2025}).
Thus, we utilize the top-down method for this case,
which was developed systematically for major modal logic in
previous papers \cite{kushida2024, kushida2025}.

Let $(L_1, \ldots, L_n)$ be the list of applications of the initial sequent or the rules
introducing the left cut formula;
for each $1\leq i\leq n$,
let $Q_i$ be the path of applications of rules consisting of
(i) $L_i$, (ii) those between $L_i$ and $I$ and (iii) $I$, in this top-down order.
Similarly,
let $(M_1, \ldots, M_m)$ be the list of applications of the initial sequent or the rules
introducing the right cut formula;
for each $1\leq j\leq m$,
let $R_j$ be the path of applications of rules
(i) $R_j$, (ii) those between $R_j$ and $I$ and (iii) $I$, in this top-down order.

When $A$ is a modal formula of the form $\Box B$,
each $L_i$ is an application of $\Box :r$; 
each $M_j$ is an application of $K$. 

When $A$ is a modal formula of the form $\blacksquare B$,
each $L_i$ is an application of $\blacksquare :r_1$,
$\blacksquare:r_2$ or the initial sequent $\blacksquare B \to \blacksquare B$;
each $M_j$ is an application of $K^{\blacksquare}$, $\blacksquare:l$
or the initial sequent $\blacksquare B \to \blacksquare B$.

When $A$ is atomic,
each $L_i$ and each $M_i$ are an application of the initial sequent.

We claim the following facts.

\vspace{1em}
{\bf Fact 1}.
If a $Q_i$ starts with $\Box:r$, $\blacksquare:r_1$ or the initial sequent
or if an $R_j$ starts with $K$, $K^{\blacksquare}$ or the initial sequent,
then we can transform $P$ so that the following hold.

{\bf Fact 1.1}. There is no application of the $\to$-rule
in the $Q_i$ nor the $R_j$.

{\bf Fact 1.2}.
When the turnstile of $P$ is $\to$,
there is no application of the $\Rightarrow$-rule
in the $Q_i$ nor the $R_j$.

{\bf Fact 1.3}.
When the turnstile of $P$ is $\Rightarrow$,
there is exactly one application of $\Rightarrow$-rule
in the $Q_i$ and the $R_j$.

{\bf Fact 2}.
If a $Q_i$ starts with $\blacksquare:r_2$
or if an $R_j$ starts with $\blacksquare:l$,
then we can transform $P$ so that the following hold.

{\bf Fact 2.1}. There is no application of the $\Rightarrow$-rule
in the $Q_i$ nor the $R_j$.

{\bf Fact 2.2}.
When the turnstile of $P$ is $\Rightarrow$,
there is no application of the $\rightarrow$-rule
in the $Q_i$ nor the $R_j$.

{\bf Fact 2.3}.
When the turnstile of $P$ is $\rightarrow$,
there is exactly one application of the rule $\rightarrow$
in the $Q_i$ and the $R_j$.

\vspace{1em}
Let us prove these facts.
Facts 1.2 and 1.3 are easily derived from Fact 1.1,
and Facts 2.2 and 2.3 are easily derived from Fact 2.1.

An occurrence of the turnstile $\Rightarrow$ in a proof $P$
is said to be a {\it main} one if it is introduced by the $\Rightarrow$-rule
and possibly by the rule $ew$;
it is said to be a {\it sub} one if it is introduced only by $ew$.

For Fact 1.1.
All these rules
($\Box:r, \blacksquare:r_1$,
the initial sequents)
have a $\to$-sequent as their lower sequent.
We consider the case when
a $Q_i$ starts with the initial sequent;
the other cases are treated similarly.
Suppose that $Q_i$
contains
an application
of the $\to$-rule
as follows.

\small

\begin{center}
\shortstack{
\shortstack{
\hspace{-5em}$A\to A$ \\
\hspace{-10em}\hspace{1.5em}$Q_i \hspace{3em}\vdots$
\vspace{1em}\\
$\infer[\Rightarrow\hspace{4em}]{\Sigma_1
\Rightarrow^{main} \Xi_1, A 
}{
\Sigma_1
\to \Xi_1, A 
}$ \vspace{1em}
\\ \vspace{1em}
\hspace{-5em} \vdots 
}
\shortstack{
$\infer[\Rightarrow]{\Upsilon_1
\Rightarrow^{main} \Psi_1
}{
\Upsilon_1
\rightarrow \Psi_1 
}$ \vspace{1em}
\\\vspace{1em}
\hspace{-2em} \vdots 
}
\\
\vspace{1em}
$\infer[\wedge:r ]{ {\cal P}_0^{\Rightarrow^{sub} }|
{\cal P}_1^{\Rightarrow^{sub} } |
\Sigma_2, \Upsilon_2
\Rightarrow^{main} \Xi_2, \Psi_2, D\wedge E, A 
}{
{\cal P}_0^{\Rightarrow^{sub} } |
\Sigma_2
\Rightarrow^{main} \Xi_2, D, A
\hspace{5em}
{\cal P}_1^{\Rightarrow^{sub} } |
\Upsilon_2
\Rightarrow^{main} \Psi_2, E
}$
\\ \vspace{1em}
\hspace{-2em} \vdots 
\\
\vspace{1em}
$
\infer[\to]{\Sigma_3
\rightarrow \Xi_3, A
}{
\Sigma_3
\Rightarrow^{main} \Xi_3, A
}
$ \\
\hspace{-4em}
$Q_i$ \hspace{1em}
\vdots
}
\end{center}

\normalsize
Here $A$ is the cut formula in $Q_i$.
Put $S=\Sigma_3
\Rightarrow^{main} \Xi_3, A$ (the upper sequent of the application of the $\to$-rule).

Let $\alpha$ denote the proof-fragment starting from the applications of rules
introducing ancestors $\Rightarrow$ of the turnstile $\Rightarrow$ of $S$
and ending with $S$.

First of all, the turnstile of $S$ is $\Rightarrow$ and its ancestors cannot coexist with
$\to$-sequents.
This is because: if the form $T^{\to}| {\cal J}^{\Rightarrow}$ occurs in $\alpha$
where $\Rightarrow$ of ${\cal J}^{\Rightarrow}$
is ancestors of the turnstile $\Rightarrow$ of $S$,
then, the turnstile $\to$ of $T^{\to}$ must change to $\Rightarrow$
and $T^{\to}| {\cal J}^{\Rightarrow}$ must eventually merge to $S$.
However, there is no applicable rule to change
$\to$ to $\Rightarrow$ in the presence of $\Rightarrow$-sequents,
that is, change the form $T^{\to}| {\cal J}^{\Rightarrow}$ to
$T^{\Rightarrow}| {\cal J}^{\Rightarrow}$.
Thus the ancestors $\Rightarrow$ of the turnstile $\Rightarrow$ of $S$
can occur only in hypersequents consisting only of $\Rightarrow$.

The turnstile $\Rightarrow$ of $S$ cannot be introduced by $\blacksquare:r_1$-rule.
This is because: the lower hypersequent of $\blacksquare:r_1$ contains both $\to$- and $\Rightarrow$-sequents.
Therefore, it is introduced by $\Rightarrow$- or $ew$-rule.
So every hypersequent
in $\alpha$
must be of the form

\vspace{1em}
${\cal P}^{\Rightarrow^{sub}}|{\cal Q}^{\Rightarrow^{main}}|
T^{\Rightarrow^{main}}$.

\vspace{1em}
\noindent
Note that it is impossible that all sequents in a hypersequent in proof $P$
are introduced by $ew$.

Note also that  there is no application of $K, 4, K^{\blacksquare}, 4^{\blacksquare}$,
$\Box:r$ nor $\blacksquare:r_1$ in $\alpha$.

\vspace{1em}
We show that
every hypersequent in $\alpha$

\vspace{1em}
${\cal P}^{\Rightarrow^{sub}}|{\cal Q}^{\Rightarrow^{main}}|
T^{\Rightarrow^{main}}$

\vspace{1em}
\noindent
can be replaced with some $\Rightarrow^{main}$-sequent in it

\vspace{1em}

$
T^{\Rightarrow^{main}}$

\vspace{1em}
\noindent
without losing its derivability from its upper sequents.
Proceed from the top of $\alpha$ to $S$.

\vspace{1em}
$\bullet$ On the rule $\Rightarrow$.

\begin{center}
$\infer[\Rightarrow]{T^{\Rightarrow^{main}}
}{
T^{\rightarrow} \hspace{1.3em}
}$
\end{center}

Obviously, the lower hypersequent satisfies the desired property.

\vspace{1em}
Below, by $T^{\Rightarrow^{main}}$ or $T_i^{\Rightarrow^{main}} (i=0,1)$, we mean the $\Rightarrow^{main}-$sequents
chosen for the upper sequents.

\vspace{1em}
$\bullet$
On the propositional rules except $\wedge:r$, the rules of $ic$, $iw$, $merge$,
$\blacksquare:l$ and $\blacksquare:r_2$ such that none of $T^{\Rightarrow^{main}}$
and $T_i^{\Rightarrow^{main}}$ is an upper sequent.
We illustrate the two cases.

\begin{center}
$\infer{
\blacksquare B \Rightarrow^{sub} |
{\cal P}^{\Rightarrow^{sub}}|
T^{\Rightarrow^{main}}
}{ B \Rightarrow^{sub} |
{\cal P}^{\Rightarrow^{sub}}| T^{\Rightarrow^{main}}
}
\hspace{3em}
\infer[merge]{
(U_0\cdot U_1)^{\Rightarrow
}|
{\cal P}^{\Rightarrow^{sub}
}|{\cal Q}^{\Rightarrow^{main}}| T^{\Rightarrow^{main}}}{
U_0^{\Rightarrow
}|U_1^{\Rightarrow
}|
{\cal P}^{\Rightarrow^{sub}}| {\cal Q}^{\Rightarrow^{main}}|T^{\Rightarrow^{main}}
}
$
\end{center}

In any case, the hypersequent to be replaced for the lower and upper hypersequents
is the same:
$
{T^{\Rightarrow^{main}}}$.
So we remove the application of these rules.

\vspace{1em}
$\bullet$ On the rule $ew$.

\begin{center}
$\infer[ew]{
U^{\Rightarrow^{sub}}|
{\cal P}^{\Rightarrow^{sub}}|{\cal Q}^{\Rightarrow^{main}}| T^{\Rightarrow^{main}}}{
\hspace{2.5em}{\cal P}^{\Rightarrow^{sub}}|{\cal Q}^{\Rightarrow^{main}}| T^{\Rightarrow^{main}}
}
$
\end{center}

As in the above case, we can remove this application, leaving $T^{\Rightarrow^{main}}$ alone.

\vspace{1em}
$\bullet$ On the rule $merge$ such that $T^{\Rightarrow^{main}}$
is an upper sequent.

\begin{center}
$\infer{
{\cal P}^{\Rightarrow^{sub}}| {\cal Q}^{\Rightarrow^{main}}|
(U\cdot T)^{\Rightarrow^{main}}
}{
{\cal P}^{\Rightarrow^{sub}}|{\cal Q}^{\Rightarrow^{main}}|
U^{\Rightarrow}| T^{\Rightarrow^{main}}
}
$
\end{center}

Here, $U^{\Rightarrow}$ may be $main$ or $sub$.
The desired property is satisfied
because
$(U\cdot T)^{\Rightarrow^{main}}$ is derivable from
$ T^{\Rightarrow^{main}}$
in terms of the rule $iw$.

\vspace{1em}
$\bullet$ On the rule $\wedge:r$ whose upper sequents are
$ T_0^{\Rightarrow^{main}}$ and $ T_1^{\Rightarrow^{main}}$.

\begin{center}
$\infer{
{\cal P}_0^{\Rightarrow^{sub}}|
{\cal P}_1^{\Rightarrow^{sub}}|
{\cal Q}_0^{\Rightarrow^{main}}|
{\cal Q}_1^{\Rightarrow^{main}}|
T^{\Rightarrow^{main}}( A\wedge B)
}{
{\cal P}_0^{\Rightarrow^{sub}}|
{\cal Q}_0^{\Rightarrow^{main}}|
T_0^{\Rightarrow^{main}}(A)
&&
{\cal P}_1^{\Rightarrow^{sub}}|
{\cal Q}_1^{\Rightarrow^{main}}|
T_1^{\Rightarrow^{main}} (B)
}
$
\end{center}

We can construct the following proof straightforwardly.

\begin{center}
$\infer{
T^{\Rightarrow^{main}}( A\wedge B)
}{
T_0^{\Rightarrow^{main}}(A)
&&
T_1^{\Rightarrow^{main}} (B)
}
$
\end{center}

$\bullet$ On the rule $\wedge:r$
such that at least one upper sequent
is not $T_0^{\Rightarrow^{main}}(A)$ or $T_0^{\Rightarrow^{main}}(B)$ .
We illustrate the case when
the left upper sequent is $U \not=T_0^{\Rightarrow^{main}}(A) $.

\begin{center}
$\infer{
{\cal P}_0^{\Rightarrow^{sub}}|
{\cal P}_1^{\Rightarrow^{sub}}|
{\cal Q}_0^{\Rightarrow^{main}}|
{\cal Q}_1^{\Rightarrow^{main}}|
T_0^{\Rightarrow^{main}}|
T^{\Rightarrow^{main}}( A\wedge B)
}{
{\cal P}_0^{\Rightarrow^{sub}}|
{\cal Q}_0^{\Rightarrow^{main}}|
T_0^{\Rightarrow^{main}}|
U^{\Rightarrow}(A)
&&
{\cal P}_1^{\Rightarrow^{sub}}|
{\cal Q}_1^{\Rightarrow^{main}}|
T_1^{\Rightarrow^{main}} (B)
}
$
\end{center}

Here, $U^{\Rightarrow}(A)$ may be a main or sub occurrence.
Since $T_0^{\Rightarrow^{main}}$ is provable,
we can replace the lower sequent with $T_0^{\Rightarrow^{main}}$.
So we can remove this application of $\wedge:r$.

\vspace{1em}
Thus, we can assume that every hypersequent in $\alpha$
is just a $\Rightarrow$-sequent.
In each path of $\alpha$,
there can be only the applications of propositional rules,
the rules of $ic, iw, \blacksquare:l, \blacksquare:r_2$.
In the path $Q_i$ in particular,
it contains the cut formula $A$
and thus there is no application of
$ \blacksquare:l$ or $ \blacksquare:r_2$ in $Q_i$.
Therefore, we can convert each hypersequent $T^{\Rightarrow}$ in $Q_i$
to $T^\to$.
When a possible application of $\wedge:r$
merges another path with $Q_i$,
the other upper sequent is converted to
a $\to$-sequent by the $\to$-rule,
as follows.

\small

\begin{center}
\shortstack{
\shortstack{\hspace{-10em}\hspace{1.5em}$Q_i \hspace{3em}\vdots$
\vspace{1em}\\
\hspace{2.3em}
$\infer[(\Rightarrow)
\hspace{6em}]{\Sigma_1
\to \Xi_1, A 
}{\Sigma_1
\to \Xi_1, A 
}$ \vspace{1em}
\\ \vspace{1em}
\hspace{-5em} \vdots 
}
\shortstack{
$\infer[\Rightarrow]{\Upsilon_1
\Rightarrow^{main} \Psi_1
}{
\Upsilon_1
\rightarrow \Psi_1 
}$ \vspace{1em}
\\\vspace{1em}
\hspace{-2em} \vdots 
}
\\
\vspace{1em}
$\infer[\wedge:r ]{ 
\Sigma_2, \Upsilon_2
\to
\Xi_2, \Psi_2, D\wedge E, A 
}{
\Sigma_2
\to
\Xi_2, D, A
\hspace{5em}
\infer[\to]{\Upsilon_2
\to
\Psi_2, E}{
\Upsilon_2
\Rightarrow^{main} \Psi_2, E
}
}
$
\\ \vspace{.5em}
\vdots 
\\
\vspace{.5em}
\hspace{2.3em}
$
\infer[(\to)]{\Sigma_3
\rightarrow \Xi_3, A
}{\Sigma_3
\rightarrow \Xi_3, A
}
$ \\
\hspace{-2.3em}
$Q_i$ \hspace{1em}\vdots
}
\end{center}

\normalsize
Finally, the application of the $\to$-rule and the accompanying application of
the $\Rightarrow$-rule have disappeared.
This finishes the proof of Fact 1.1.

\vspace{1em}
For Fact 1.2.
When the turnstile of $P$ is $\to$,
if such a $Q_i$ or an $R_j$ contains an application of $\Rightarrow$,
the introduced turnstile $\Rightarrow$ changes to $\to$,
and the change is made by the rule $\to$, $\Box:r$ or $\blacksquare:r_1$.
By Fact 1.1, there is no possibility of the $\to$-rule.
Also, there is no possibility of
$\Box:r$ or $\blacksquare:r_1$,
due to the presence of the cut formula $A$
in the $\Rightarrow$-sequent.
Thus, there is no $\Rightarrow$-rule in the $Q_i$ nor the $R_j$.

\vspace{1em}
For Fact 1.3.
When the turnstile of $P$ is $\Rightarrow$,
if such a $Q_i$ or an $R_j$
contains two applications of the
$\Rightarrow$-rule,
the turnstile $\Rightarrow$ introduced by the first application
changes to $\to$ above the second one.
This is impossible for the same reason as in Fact 1.2.

\vspace{1em}
For Fact 2.1.
The lower sequents of the rules involved here
($\blacksquare:r_2$, $\blacksquare:l$)
are $\Rightarrow$-sequents.
Consider $R_j$ starting with
$\blacksquare:l$.
(The case for $Q_i$ with
$\blacksquare:r_2$ is handled
precisely in the same way.)
Suppose that $R_j$ contains an application of the $\Rightarrow$-
rule,
as in the figure below.

\small

\begin{center}
\shortstack{
\shortstack{\hspace{-3.5em}
\infer[\blacksquare:l]{{\cal H}_0 | \blacksquare B
\Rightarrow }{
{\cal H}_0 | B
\Rightarrow
} \\
\hspace{-10em}\hspace{1.5em}$R_j \hspace{3em}\vdots$
\vspace{1em}\\
\hspace{1em}
$\infer[\to\hspace{5.5em}]{
\blacksquare B, \Sigma_1
\rightarrow^\tau \Xi_1 
}{
\blacksquare B, \Sigma_1
\Rightarrow \Xi_1 
}$ \vspace{1em}
\\ \vspace{1em}
\hspace{-5em} \vdots 
}
\shortstack{
$\infer[K]{{\cal J}_0|\Box F \to |\Upsilon_1
\Rightarrow \Psi_1
}{{\cal J}_0|
F, \Upsilon_1
\Rightarrow \Psi_1 
}$ \vspace{1em}
\\\vspace{1em}
\hspace{-2em} \vdots 
}
\\
\vspace{1em}
$\infer[\wedge:r ]{ {\cal H}_2^{\Rightarrow }| {\cal I}_2^{\rightarrow }|
{\cal J}_2^{\Rightarrow }| {\cal K}_2^{\rightarrow }|
\blacksquare B,\Sigma_2, \Upsilon_2
\rightarrow^\tau \Xi_2, \Psi_2, D\wedge E
}{
{\cal H}_2^{\Rightarrow }| {\cal I}_2^{\rightarrow }|
\blacksquare B, \Sigma_2
\rightarrow^\tau \Xi_2, D
\hspace{5em}
{\cal J}_2^{\Rightarrow }| {\cal K}_2^{\rightarrow }|
\Upsilon_2
\rightarrow^\tau \Psi_2, E
}$
\\ \vspace{1em}
\hspace{-2em} \vdots 
\\ \vspace{.6em}
\hspace{-2em} \vdots 
\\
\vspace{1em}
\hspace{-1.4em}
$
\infer[\Rightarrow]{\blacksquare B, \Sigma_3
\Rightarrow \Xi_3
}{
\blacksquare B, \Sigma_3
\rightarrow^\tau \Xi_3
}
$
\vspace{-1em}\\
\hspace{-4em}
$R_j$ \hspace{1em}\vdots
}
\end{center}

\normalsize
Put $S=\blacksquare B, \Sigma_3
\rightarrow \Xi_3$ (the upper sequent of the application).
The turnstile $\to$ of $S$ must be introduced by the $\to$-rule in $R_j$.
(It cannot be introduced by any other rule in $R_j$
because of the presence of $\blacksquare B$;
it can be introduced, for example, by the $K$-rule
in another path merging with $R_j$
as the above figure illustrates.)
Let 
$T$ denote the lower sequent
of the application of the $\to$-rule: $\blacksquare B, \Sigma_1
\rightarrow \Xi_1$.

We are going to change the turnstile $\to^{\tau}$ labeled with $\tau$ to $\Rightarrow$.
The turnstile $\to^{\tau}$
is that of $T$ and its descendants (including $S$) in $R_j$,
and the turnstile of
another upper sequent of possible
applications of $\wedge:r$
merging with $R_j$.


We show that every hypersequent in $R_j$

\vspace{1em}
${\cal H}_i^{\Rightarrow}|{\cal I}_i^{\rightarrow}| \blacksquare B, \Sigma_i
\rightarrow \Xi_i$

\vspace{1em}
\noindent
can be replaced with

\vspace{1em}
$
(\cdot({\cal I}_i^{\rightarrow})) \cdot (\blacksquare B, \Sigma_i
\rightarrow \Xi_i)$

\vspace{1em}
\noindent
without losing its derivability from its upper sequents.
 
Proceed from $T$ to $S$.
Note that we do not go through the whole proof-fragment
as in the proof of Fact 1.1; we are just going through the path $R_j$.

As the base case,
$T$ clearly satisfies the desired property.

\vspace{1em}

$\bullet$ On $\wedge:r$ for
the $\to$-sequents,
it is the situation illustrated above.
On the left upper hypersequent
of the application of $\wedge:r$,
we can suppose that
$ 
(\cdot{\cal I}_2^{\rightarrow })\cdot
(\blacksquare B, \Sigma_2
\rightarrow^\tau \Xi_2, D)$
is provable.

\begin{center}
$\infer[\wedge:r ]{
{\cal J}_2^{\Rightarrow }| {\cal K}_2^{\rightarrow }|
(\cdot{\cal I}_2^{\rightarrow })\cdot( \blacksquare B,\Sigma_2, \Upsilon_2
\rightarrow^\tau \Xi_2, \Psi_2, D\wedge E)
}{
(\cdot{\cal I}_2^{\rightarrow })\cdot
(\blacksquare B, \Sigma_2
\rightarrow^\tau \Xi_2, D)
\hspace{2.5em}
{\cal J}_2^{\Rightarrow }| {\cal K}_2^{\rightarrow }|
\Upsilon_2
\rightarrow \Psi_2, E
}$
\end{center}

The $\Rightarrow$-sequents
in $ {\cal J}_2^{\Rightarrow }$
eventually turn into
$\to$-sequents and merged to
the $\to$-sequent, $S$.
The transformation into $\to$-sequents
must be made according to the rule of
$\Box:r$ or $\blacksquare:r_1$,
because the presence of the $\to^\tau$-sequent makes the $\to$-rule unapplicable.
However, by Lemma \ref{standard},
we can assume that
the application of $\Box:r$
or $\blacksquare:r_1$ is of the standard form,
and the presence of the application of $\wedge:r$ above that
application of
$\Box:r$
or $\blacksquare:r_1$
breaks the condition of the standard form.
Thus, we may assume that
${\cal J}_2^{\Rightarrow}$
is empty.
After all, we can have the desired derivation as follows.

\begin{center}
$\infer[\wedge:r ]{
(\cdot {\cal K}_2^{\rightarrow })
\cdot
(\cdot{\cal I}_2^{\rightarrow })\cdot( \blacksquare B,\Sigma_2, \Upsilon_2
\rightarrow^\tau \Xi_2, \Psi_2, D\wedge E)
}{
(\cdot{\cal I}_2^{\rightarrow })\cdot
(\blacksquare B, \Sigma_2
\rightarrow^\tau \Xi_2, D)
\hspace{2.5em}
\infer={ (\cdot{\cal K}_2^{\rightarrow })\cdot(
\Upsilon_2
\rightarrow \Psi_2, E)}{
{\cal K}_2^{\rightarrow }|
\Upsilon_2
\rightarrow \Psi_2, E
}
}$
\end{center}

$\bullet$ On the propositional rules except $\wedge:r$ for $\to$-sequents,
and the rules $ic$, $iw$
for $\to$-sequents, they can be simulated in the sequents of $\to^\tau$.

\begin{center}
$
\infer{{\cal H}_i^{\Rightarrow}|\neg C, \Sigma'_i
\rightarrow \Xi'_i |
{\cal I}_i^{\rightarrow}| \blacksquare B, \Sigma_i
\rightarrow^\tau \Xi_i}{{\cal H}_i^{\Rightarrow}|\Sigma'_i
\rightarrow \Xi'_i, C|
{\cal I}_i^{\rightarrow}| \blacksquare B, \Sigma_i
\rightarrow^\tau \Xi_i}
$

\vspace{1em}
$\bigtriangledown$

\vspace{1em}
$
\infer{(\neg C, \Sigma'_i
\rightarrow \Xi'_i )\cdot
(\cdot{\cal I}_i^{\rightarrow})\cdot (\blacksquare B, \Sigma_i
\rightarrow^\tau \Xi_i)}{
(\Sigma'_i
\rightarrow \Xi'_i, C )\cdot
(\cdot{\cal I}_i^{\rightarrow})\cdot (\blacksquare B, \Sigma_i
\rightarrow^\tau \Xi_i)}
$
\end{center}

$\bullet$ On $\wedge:r$ for $\Rightarrow$-sequents,
the desired proof is obtained as follows.

\begin{center}
$\infer[\wedge:r ]{ {\cal H}_2^{\Rightarrow }| {\cal I}_2^{\rightarrow }|
{\cal J}_2^{\Rightarrow }| {\cal K}_2^{\rightarrow }|
\Sigma'_2, \Upsilon_2
\Rightarrow \Xi'_2, \Psi_2, D\wedge E
|
\blacksquare B, \Sigma_2
\rightarrow^\tau \Xi_2
}{
{\cal H}_2^{\Rightarrow }|
{\cal I}_2^{\rightarrow }|
\Sigma'_2
\Rightarrow \Xi'_2, D|
\blacksquare B, \Sigma_2
\rightarrow^\tau \Xi_2
\hspace{5em}
{\cal J}_2^{\Rightarrow }
| {\cal K}_2^{\rightarrow }
|
\Upsilon_2
\Rightarrow \Psi_2, E
}$

\vspace{1em}
$\bigtriangledown$

\vspace{1em}

$\infer=[iw ]{
(\cdot{\cal K}_2^{\rightarrow })\cdot
(\cdot{\cal I}_2^{\rightarrow })\cdot( \blacksquare B,\Sigma_2
\rightarrow^\tau \Xi_2)
}{
(\cdot{\cal I}_2^{\rightarrow })\cdot
(\blacksquare B, \Sigma_2
\rightarrow^\tau \Xi_2)
}$
\end{center}

\vspace{1em}
$\bullet$ On the propositional rules and the rules $ic$, $iw$, $merge$,
and $ew$ for $\Rightarrow$-sequents,
they are just deleted.

\vspace{1em}
$\bullet$ On the $K$-, $4$-, $K^\blacksquare$- or $4^\blacksquare$-rule,
the desired proof is obtained as follows.

\begin{center}
$\infer[\wedge:r ]{ {\cal H}_2^{\Rightarrow }| {\cal I}_2^{\rightarrow }|
\blacksquare D \to |
\Sigma'_2
\Rightarrow \Xi'_2
|
\blacksquare B, \Sigma_2
\rightarrow^\tau \Xi_2
}{
{\cal H}_2^{\Rightarrow }|
{\cal I}_2^{\rightarrow }|
\blacksquare D, \Sigma'_2
\Rightarrow \Xi'_2|
\blacksquare B, \Sigma_2
\rightarrow^\tau \Xi_2
}$

\vspace{1em}
$\bigtriangledown$

\vspace{1em}

$\infer=[iw ]{
(\cdot{\cal I}_2^{\rightarrow })
\cdot(\blacksquare D \to )
\cdot( \blacksquare B,\Sigma_2
\rightarrow^\tau \Xi_2)
}{
(\cdot{\cal I}_2^{\rightarrow })\cdot
(\blacksquare B, \Sigma_2
\rightarrow^\tau \Xi_2)
}$
\end{center}

Thus, now we can assume that there are only applications of propositional rule,
$iw$ and $ic$ on $\to^\tau$-sequents 
in $R_j$.
Then, we can convert the turnstile of $\to^\tau$ to $\Rightarrow$
to obtain the following proof,
where those applications of
the $\to$- and $\Rightarrow$-rules
disappear.

\small
\begin{center}
\shortstack{
\shortstack{\hspace{-3.5em}
\infer[\blacksquare:l]{{\cal H}_0 | \blacksquare B
\Rightarrow }{
{\cal H}_0 | B
\Rightarrow
} \\
\hspace{-10em}\hspace{1.5em}$R_j \hspace{3em}\vdots$
\vspace{1em}\\ \hspace{1em}
$\infer[(\to)\hspace{4em}]{
\blacksquare B, \Sigma_1
\Rightarrow \Xi_1 
}{
\blacksquare B, \Sigma_1
\Rightarrow \Xi_1 
}$ \vspace{1em}
\\ \vspace{1em}
\hspace{-5em} \vdots 
}
\shortstack{
$\infer[K]{{\cal J}_0|\Box F \to |\Upsilon_1
\Rightarrow \Psi_1
}{{\cal J}_0|
F, \Upsilon_1
\Rightarrow \Psi_1 
}$ \vspace{1em}
\\\vspace{1em}
\hspace{-2em} \vdots 
}
\\
\vspace{1em}
$\infer[\wedge:r ]{ (\cdot{\cal I}_2^{\Rightarrow })\cdot
(\cdot{\cal K}_2^{\Rightarrow })\cdot
( \blacksquare B,\Sigma_2, \Upsilon_2
\Rightarrow \Xi_2, \Psi_2, D\wedge E)
}{
(\cdot{\cal I}_2^{\Rightarrow })\cdot
(\blacksquare B, \Sigma_2
\Rightarrow \Xi_2, D)
\hspace{5em}
\infer[\Rightarrow]{
(\cdot {\cal K}_2^{\Rightarrow })
\cdot(\Upsilon_2
\Rightarrow \Psi_2, E)
}{
\infer=[merge]{
(\cdot {\cal K}_2^{\rightarrow })
\cdot(\Upsilon_2
\rightarrow \Psi_2, E)
}{
{\cal K}_2^{\rightarrow }|
\Upsilon_2
\rightarrow \Psi_2, E
}
}
}$
\\ \vspace{1em}
\hspace{-2em} \vdots 
\\ \vspace{.6em}
\hspace{-2em} \vdots 
\\
\vspace{1em}
$
\infer[(\Rightarrow)]{\blacksquare B, \Sigma_3
\Rightarrow \Xi_3
}{\blacksquare B, \Sigma_3
\Rightarrow \Xi_3
}
$ \vspace{-.5em}\\
\hspace{-4.4em}
$R_j$ \hspace{1.5em}\vdots
}
\end{center}

\normalsize

\vspace{1em}
For Fact 2.2.
Suppose that the turnstile of $P$ is $\Rightarrow$,
and
the $Q_i$ or the $R_j$ contains an application of $\rightarrow$.
The introduced $\to$-sequent containing the cut formula $\blacksquare B$
eventually turns into a $\Rightarrow$-sequent.
The transformation must be made by the $\Rightarrow$-rule.
This is, however, impossible in the transformed proof of Fact 2.1.

\vspace{1em}
For Fact 2.3.
Suppose that the turnstile of $P$ is $\rightarrow$,
and there are two applications of the $\to$-rule in the $Q_i$ or the $R_j$.
Then, the $\to$-sequent introduced by the first one eventually turns into a
$\Rightarrow$-sequent by the $\Rightarrow$-rule.
This never happens in the transformed proof of { Fact 2.1}.

This completes the proof of the facts.

\vspace{1em}
Let us begin with the \noindent {\it Case 1} of Theorem \ref{cut-elimination}.

\vspace{1em}
\noindent {\it Case 1}. $A=\Box B$.
Suppose that $P$ is as follows.

\begin{center}
$
\infer{{\cal H} | {\cal I} |
\Gamma, \Pi \gg \Delta, \Theta
}{
\shortstack{\deduc \vspace{.6em}
\\
$\infer[I]{{\cal H} |\Gamma \gg \Delta, \Box B}{}$
} &
\shortstack{\deduc \vspace{.6em}\\
$\infer[J]{ {\cal I} |\Box B, \Pi \gg \Theta}{}$}
}
$
\end{center}


\noindent {\it Case 1.1}.
The turnstile of $P$ is $\to$.
$P$ is now as follows.

\small
\begin{center}
\shortstack{
$\infer[L_i]{\nabla \Lambda_i\to |
\rightarrow \Box B}{
\nabla \Lambda_i\to |
\Box B\Rightarrow B
}$ $\hspace{2em}$
$\infer[L_j]{
{\cal K}| 
\to \Box B
}{ \infer={
{\cal K}|
\Box B \Rightarrow B
}{
{\cal K}| \Rightarrow}
}
$
$\hspace{2em}$
$\infer[M_k]
{{\cal L}|
\Box B \to
| \Sigma_k \Rightarrow \Psi_k
}
{{\cal L}|
B, \Sigma_k \Rightarrow \Psi_k
}$
\\ \vspace{1em}
\hspace{-3em}
$Q_i$ \hspace{1em} $\vdots$
\hspace{7em}
$Q_j$ \hspace{1em} $\vdots$
\hspace{7em}
$R_k$ \hspace{1em} $\vdots$
\\
$\vdots$
\hspace{9.6em} 
$\vdots$
\hspace{8.5em}
\hspace{1em} $\vdots$
\\ \vspace{1em}
\\
\hspace{.4em}
$\ddots
\hspace{7.5em}\iddots
\hspace{10.5em}
\vdots
$
\vspace{3em}
\\
\hspace{5.5em}
$
\infer{{\cal H} | {\cal I} |
\Gamma, \Pi \rightarrow \Delta, \Theta
}{
\infer[I]{{\cal H} |\Gamma \to \Delta, \Box B}{}
\hspace{3.7em}& \hspace{3.7em}
\infer[J]{ {\cal I} |\Box B, \Pi \rightarrow \Theta}{}
}
$
}
\end{center}

\normalsize

\vspace{1em}

There are three types of paths,
$Q_i, Q_j, R_k$ as above.
By Lemma \ref{standard}, 
we may assume that $L_i$ and $L_j$ are of the standard form.

Note that, by { Facts 1.1, 1.2},
there is no application of $\Rightarrow$- or $\to$-rule
in $Q_i$, $Q_j$ and $R_k$.

Also, by Corollary \ref{crucial_corollary},
 the diagonal formula is now deleted.
Thus, in $L_i$, \begin{center}
$ \nabla \Lambda_i\to |
\Rightarrow
B$
\end{center}
is provable,
where in turn extra applications of $cut$ with $B$ the cut formula 
can occur.

Now we work on the top-down method.
A rough sketch of the
method is as follows.
For the top of each of $Q_i$ and $Q_j$:

\begin{center}
$ \nabla \Lambda_i\to| \to \Box B$
\hspace{1.5em}
and
\hspace{1.5em}
${\cal K}| 
\to \Box B$,
\end{center}

\noindent
we are going to
replace them with the following,
respectively.

\begin{center}
${\cal I} |
\nabla \Lambda_i\to| \Pi \to \Theta$
\hspace{1.5em}
and
\hspace{1.5em}
${\cal I} |
{\cal K}| 
\Pi \to \Theta$.
\end{center}

Here, 
${\cal I} |
\Pi \to \Theta$
comes from the lower  hypersequent of $J$.

To do this replacement,
we execute the $cut$ for
the subformula $B$
in the former case.
Then, under this replacement,
we simulate
the subproof above $I$
to obtain the original end
hypersequent.

Below is a detailed
description of the proof-transformation of the method.

\vspace{1em}
$\bullet$ For $Q_i$, 
we make use of
each $R_k$ to obtain
${\cal I}|\nabla \Lambda_i \to| \Pi\rightarrow \Theta$ as follows.
Firstly, for each $R_j$, we do the following replacement.

\begin{center}
$
\infer[M_k\hspace{2em} \rhd \hspace{2em}]
{{\cal L}|
\Box B \to
| \Sigma_k \Rightarrow \Psi_k
}
{
{\cal L}|
B, \Sigma_k \Rightarrow \Psi_k
}
\infer[cut]
{
{\cal L}|
\nabla \Lambda_i\to |\Sigma_k \Rightarrow \Psi_k
}
{ \nabla \Lambda_i\to |
\Rightarrow B \hspace{1.5em}
{\cal L}|
B, \Sigma_k \Rightarrow \Psi_k
}
$
\end{center}

Secondly, under this replacement,
we simulate the subproof above $J$.
Then, we obtain: ${\cal I}| \nabla \Lambda_i, \Pi \to \Theta
$ and thus ${\cal I}| \nabla \Lambda_i\to | \Pi \to \Theta
$.
The following illustrates this proof-transformation.

\begin{center}
\shortstack{
$
{\cal L}|
\Box B \to
| \Sigma_k \Rightarrow \Psi_k
\hspace{9em}
{\cal L}|
\nabla \Lambda_i\to | \Sigma_k \Rightarrow \Psi_k
$
\vspace{1em}
\\
\hspace{-1em}
\deduc
\hspace{6em} $\rhd$ \hspace{6em} \deduc
\\
\vspace{1em}
\\
\hspace{3em}
$
{\cal I}|\Box B, \Pi \rightarrow \Theta
\hspace{11em}
\infer[split]{ {\cal I} | \nabla \Lambda_i\to| \Pi
\rightarrow \Theta}{
{\cal I} | \nabla \Lambda_i, \Pi
\rightarrow \Theta
}
$
}
\end{center}

Thus, we can make the replacement
for the top of $Q_i$.

\vspace{1em}
$\bullet$ For $Q_j$,
firstly, for each $R_k$, we do the following replacement similar to the former case,
although we do not use the $cut$-rule.

\begin{center}
$
\infer[M_k\hspace{2em} \rhd \hspace{2em}]
{{\cal L}|
\Box B \to
| \Sigma_k \Rightarrow \Psi_k
}
{
{\cal L}|
B, \Sigma_k \Rightarrow \Psi_k
}
\infer={
{\cal K}| {\cal L}| \to |
\Sigma_k \Rightarrow \Psi_k
}{{\cal K}|\Rightarrow
}
$
\end{center}

Secondly, under the replacement,
we make a simulation of
the subproof above $J$ as in the former case, as illustrated below.

\begin{center}
\shortstack{
$
{\cal L}|
\Box B \to
| \Sigma_k \Rightarrow \Psi_k
\hspace{8.5em}
{\cal K}| {\cal L}|
\to |\Sigma_k \Rightarrow \Psi_k
$
\vspace{1em}
\\
\deduc
\hspace{6em} $\rhd$ \hspace{5em} \deduc
\\
\vspace{1em}
\\
$
{\cal I}|\Box B, \Pi \rightarrow \Theta
\hspace{11em}
{\cal K}| {\cal I} | 
\Pi
\rightarrow \Theta
$
}
\end{center}

Then, under the replacement for $Q_i$ and $Q_j$,
we simulate the subproof above $I$ to obtain the original end hypersequent
as the following illustration shows.

\small

\begin{center}
\shortstack{
${\cal I}| \nabla \Lambda_j\to| \Pi \to \Theta$
$\hspace{2em}$
${\cal K}| {\cal I} | 
\Pi
\rightarrow \Theta$
\\ \vspace{1em}
\hspace{-3em}
$Q_i$ \hspace{1em} $\vdots$
\hspace{7em}
$Q_j$ \hspace{1em} $\vdots$
\\
$\vdots$
\hspace{9.6em} 
$\vdots$
\\ \vspace{1em}
\\
$\ddots
\hspace{7.5em}\iddots
$
\vspace{2em}
\\
\hspace{.5em}
$
\infer[I]{{\cal H} |{\cal I} |\Gamma, \Pi \to \Delta, \Theta}{}
$
}
\end{center}

\normalsize
Basically, the proof-transformation in the following 
cases ({\it Cases 1.2, 2.1, 2.2 })
is similar to that in \noindent {\it Case 1.1},
though some extra technical modifications are needed for each of the coming cases.

\noindent {\it Case 1.2}.
The turnstile of $P$ is $\Rightarrow$.
$P$ is now as follows.

\footnotesize
\begin{center}
\shortstack{
$\infer[L_i]{\nabla \Lambda_i\to |
\rightarrow \Box B}{
\nabla \Lambda_i\to |
\Box B\Rightarrow B
}$ $\hspace{2em}$
$\infer[L_j]{
{\cal K}| 
\to \Box B
}{
\infer={
{\cal K}|
\Box B \Rightarrow B
}{
{\cal K}| \Rightarrow 
}
}$
$\hspace{2em}$
$\infer[M_k]
{{\cal L}|
\Box B \to
| \Sigma_k \Rightarrow \Psi_k
}
{{\cal L}|
B, \Sigma_k \Rightarrow \Psi_k
}$
\\
\hspace{-6em}
$Q_i$ \hspace{1em} $\vdots$
\hspace{8em}
$Q_j$ \hspace{1em} $\vdots$
\hspace{9em}
$R_k$ \hspace{1em} $\vdots$
\\
\hspace{-3.2em}
$\vdots$
\hspace{10.6em} 
$\vdots$
\hspace{10.4em}
\hspace{1em} $\vdots$
\vspace{1em}
\\
\hspace{-1.5em}
\infer{\Xi_i\Rightarrow\Upsilon_i, \Box B}{\Xi_i\rightarrow\Upsilon_i, \Box B}
\hspace{5em}
\infer{\Xi_j\Rightarrow\Upsilon_j, \Box B}{\Xi_j\rightarrow\Upsilon_j, \Box B}
\hspace{6.5em}
\infer{\Box B, \Xi_k\Rightarrow\Upsilon_k}{\Box B, \Xi_k\rightarrow\Upsilon_k}
\\
\hspace{-3em}
$\vdots$
\hspace{10.6em} 
$\vdots$
\hspace{10.5em}
\hspace{1em} $\vdots$
\\ \vspace{1em}
\\
\hspace{-2em}
$\ddots
\hspace{7.5em}\iddots
\hspace{13.5em}
\vdots
$
\vspace{2em}
\\
\hspace{2.5em}
$
\infer{{\cal H} | {\cal I} |
\Gamma, \Pi \Rightarrow \Delta, \Theta
}{
\infer[I]{{\cal H} |\Gamma \Rightarrow \Delta, \Box B}{}
\hspace{5em}& \hspace{5em}
\infer[J]{ {\cal I} |\Box B, \Pi \Rightarrow \Theta}{}
}
$
}
\end{center}

\normalsize
There are three types of paths,
$Q_i, Q_j, R_k$, as above.
Note that, by { Facts 1.1, 1.3},
there is no application of $\to$-rule and
exactly one application of $\Rightarrow$-rule
in $Q_i$, $Q_j$, and $R_k$.

\normalsize
For the application of the $\Rightarrow$-rule in $Q_i$ and $Q_j$,
we are going to replace their lower hypsesequent $\Xi_{(i,j)}\Rightarrow\Upsilon_{(i,j)}, \Box B$
with ${\cal I}|(\Xi_{(i,j)}\Rightarrow\Upsilon_{(i,j)})\cdot(\Pi\Rightarrow \Theta)$.

\vspace{1em}
$\bullet$ For $Q_i$,
first,
for each $M_k$ (the top of $R_k$),
we make precisely the same replacement as in \noindent {\it Case 1.1}
using the $cut$ with the formula $B$.
Then, simulate $R_k$, $Q_i$ and $R_k$, using the $split$-rule, as follows.

\small
\begin{center}
\shortstack{
$
\infer[cut]{
{\cal L}| \nabla \Lambda_i\to | \Sigma_k \Rightarrow \Psi_k
}{ \nabla \Lambda_i\to |
\Rightarrow B \hspace{1.5em}
{\cal L}| B, \Sigma_k \Rightarrow \Psi_k
}
$ 
\\ \hspace{-4em}
$R_k$ \hspace{1em} $\vdots$
\vspace{.5em}
\\
$\infer[split]{\nabla \Lambda_i \to | \Xi_k\rightarrow\Upsilon_k}{
\nabla \Lambda_i, \Xi_k\rightarrow\Upsilon_k}$
\\
\hspace{-4.5em}
$ Q_i$ \hspace{1em} $\vdots$
\\
\vspace{.5em}
\\
\hspace{-2em}
$\Xi_i, \Xi_k \Rightarrow\Upsilon_i, \Upsilon_k$
\\ \vspace{1em}
\hspace{-4.5em}
$R_k$ \hspace{1em} $\vdots$
\\
\hspace{-1em}
$
\infer[J]{{\cal I}|\Xi_i, \Pi \Rightarrow\Upsilon_i, \Theta}{}
$
}
\end{center}

\normalsize
$\bullet$ For $Q_j$.
First,
for each $M_k$ (the top of $R_k$),
by using the provable hypersequent ${\cal K}| \Rightarrow$,
we construct the following.

\begin{center}
$\infer={
{\cal K}| {\cal L}|
\to |
\Sigma_k \Rightarrow \Psi_k
}{{\cal K}|
\Rightarrow }$
\end{center}

Then, simulate $R_k$, $Q_j$ and $R_k$
to obtain:
${\cal I}|\Xi_j, \Pi \Rightarrow\Upsilon_j, \Theta$.

\small

\begin{center}
\shortstack{
$\infer={
{\cal K}| {\cal L}|
\to |
\Sigma_k \Rightarrow \Psi_k
}{{\cal K}|
\Rightarrow }$
\\ \hspace{-3em}
$R_k$ \hspace{1em} $\vdots$
\vspace{.5em}
\\
$
{\cal K}| 
\Xi_k \rightarrow \Upsilon_k
$
\\
\hspace{-3.5em}
$ Q_j$ \hspace{1em} $\vdots$
\\
\vspace{.5em}
\\
$\Xi_j, \Xi_k \Rightarrow\Upsilon_j, \Upsilon_k$
\\ \vspace{1em}
\hspace{-3.5em}
$R_k$ \hspace{1em} $\vdots$
\\
$
\infer[J]{{\cal I}|\Xi_j, \Pi \Rightarrow\Upsilon_j, \Theta}{}
$
}
\end{center}

\normalsize
\vspace{1em}
In this way, for $Q_i$ and $Q_j$,
we obtain the hypersequent:

\begin{center}
${\cal I}|(\Xi_{(i,j)}\Rightarrow\Upsilon_{(i,j)})\cdot(\Pi\Rightarrow \Theta)$.
\end{center}

Finally, under the replacement for the top of $Q_i$ and $Q_j$, we can simulate
the subproof above $I$ (below the applications of the $\Rightarrow$-rule)
to obtain the original end hypersequent
as follows.

\small

\begin{center}
\shortstack{
${\cal I}|(\Xi_{i}\Rightarrow\Upsilon_{i})\cdot(\Pi\Rightarrow \Theta)$
$\hspace{2em}$
${\cal I}|(\Xi_{j}\Rightarrow\Upsilon_{j})\cdot(\Pi\Rightarrow \Theta)$
\\ \vspace{1em}
\hspace{-3em}
$Q_i$ \hspace{1em} $\vdots$
\hspace{7em}
$Q_j$ \hspace{1em} $\vdots$
\\
$\vdots$
\hspace{9.6em} 
$\vdots$
\\ \vspace{1em}
\\
$\ddots
\hspace{7.5em}\iddots
$
\vspace{2em}
\\
\hspace{.5em}
$
\infer[I]{{\cal H} |{\cal I} |\Gamma, \Pi \to \Delta, \Theta}{}
$
}
\end{center}

\normalsize

\vspace{1em}
\noindent {\it Case 2}. $A$ is a modal formula $\blacksquare B$.

\noindent {\it Case 2.1}.
The turnstile of $P$ is $\to$.
First we show: if any $Q_i$ starts with $\blacksquare:r_2$
or
if any $R_j$ starts with $\blacksquare:l$,
then the $cut$ is easily eliminated.

\small
\begin{center}
\shortstack{
$\infer[L_i]{\nabla \Lambda_i\to |
\Rightarrow \blacksquare B}{
\nabla \Lambda_i\to |\Rightarrow B
}$ 
\\ \hspace{-4em}
$Q_i$ \hspace{1em} $\vdots$
\\
\hspace{-1.2em}
$\vdots$
\vspace{1em}
\\
\infer[\to
]{\Xi_i\rightarrow\Upsilon_i, \blacksquare B}{\Xi_i\Rightarrow\Upsilon_i, \blacksquare B}
\\ \vspace{1em}
\hspace{-1em}
$\vdots$
\\ \vspace{1em}
\hspace{-1em}
$\vdots$
\\
$
\infer[I]{{\cal H} |\Gamma \rightarrow \Delta, \blacksquare B}{}
$
}
\hspace{2em}
\shortstack{
\hspace{1.5em}
$\infer[M_j]{
\nabla \Lambda_k \to |\blacksquare B \Rightarrow
}{
\nabla \Lambda_k \to | B \Rightarrow
}$ 
\\ \hspace{-4em}
$R_j$ \hspace{1em} $\vdots$
\\
\hspace{-1.2em}
$\vdots$
\vspace{1em}
\\
\hspace{2em}
\infer[\to]{\blacksquare B, \Xi_k\rightarrow\Upsilon_k}{\blacksquare B, \Xi_k\Rightarrow\Upsilon_k}
\\ \vspace{1em}
\hspace{-1em}
$\vdots$
\\ \vspace{1em}
\hspace{-1em}
$\vdots$
\\
\hspace{1em}
$
\infer[J]{{\cal I} |\blacksquare B, \Delta \rightarrow \Pi}{}
$
}
\end{center}

\normalsize

By {Facts 2.1, 2.3},
such a $Q_i$ and such an $R_j$ contain no application of $\Rightarrow$-rule and
exactly one application of the $\to$-rule.
Then, $\nabla \Lambda_{i}$ and $\nabla \Lambda_{ k}$ must be empty,
because there is no applicable rule to change the turnstile of $\nabla \Lambda_{(i,k)} \to$
in the presence of the sequent $\blacksquare B\Rightarrow$ or $\Rightarrow \blacksquare B$.
Therefore, $\Rightarrow B$ and $B \Rightarrow$ are provable in the two cases.

$\bullet$ When $\Rightarrow B$ is provable,
we can use the subproof above $J$
to obtain
${\cal I} |\Delta \rightarrow \Pi$
and thus
${\cal H}|{\cal I} |\Gamma, \Delta \rightarrow \Delta, \Pi$ as desired.
Here are details. When an $R_j$ starts with $K^{\blacksquare}$,
we can use the $cut$-rule to delete $B$ with $\Rightarrow B$.

\begin{center}

$\infer[M_k \hspace{2em}\rhd\hspace{2em}]
{{\cal L}|
\blacksquare B \to
| \Sigma_k \Rightarrow \Psi_k
}
{{\cal L}| B,
\Sigma_k \Rightarrow \Psi_k
}$
$\infer[cut]
{
{\cal L}|
\Sigma_k \Rightarrow \Psi_k
}
{\Rightarrow B \hspace{2em}{\cal L}|
B, \Sigma_k \Rightarrow \Psi_k
}$
\end{center}

When an $R_j$ starts with ${\blacksquare}B \to \blacksquare B$,
we can have $\to \blacksquare B$ as follows.

\begin{center}

$\infer[M_k \hspace{2em}\rhd\hspace{2em}]
{{\blacksquare}B \to \blacksquare B
}
{
}$
$\infer[\to]{
\to \blacksquare B}{
\infer[\blacksquare:r_2]{\Rightarrow \blacksquare B}{
\Rightarrow B
}
}$
\end{center}

Then, under this replacement, we can simulate the subproof above $J$,
to obtain
${\cal I} |\Delta \rightarrow \Pi$.

$\bullet$ When $B \Rightarrow$ is provable,
the transformation is done in a symmetrical way as the above case.

\vspace{1em}
Now we can assume that there is no $Q_i$ or $R_j$
that
starts with $\blacksquare:r_2$
or $\blacksquare:l$.
Then, $P$ is as follows.

\footnotesize

\begin{center}
\hspace{2em}
\shortstack{
$\infer[L_i]{\nabla \Lambda_i\to |
\rightarrow \blacksquare B| \Rightarrow}{
\nabla \Lambda_i\to |
\Box B\Rightarrow B
}$
$\hspace{1.7em}$
$\infer[L_j]{\blacksquare B
\rightarrow \blacksquare B }{
}$
\hspace{3em}
$\infer[M_k]
{{\cal L}|
\blacksquare B \to
| \Sigma_k \Rightarrow \Psi_k
}
{{\cal L}|
B, \Sigma_k \Rightarrow \Psi_k
}$
\hspace{2em}
$\infer[M_l]{
\blacksquare B \rightarrow \blacksquare B
}{}$
\\
\hspace{-4em}
$Q_i$ \hspace{1em} $\vdots$
\hspace{6em}
$Q_j$ \hspace{1em} $\vdots$
\hspace{8em}
$R_k$ \hspace{1em} $\vdots$
\hspace{8em}
$R_l$ \hspace{1em} $\vdots$
\\
\hspace{-1.2em}
$\vdots$
\hspace{8.7em} 
$\vdots$
\hspace{10.7em} 
$\vdots$
\hspace{10.5em} 
$\vdots$
\vspace{1em}
\\
\hspace{-3em}
$\ddots
\hspace{4.5em}\iddots
$
\hspace{14em}
$\ddots
\hspace{4.5em}\iddots
$
\vspace{2em}
\\
\hspace{-2em}
$
\infer[I]{{\cal H} |\Gamma \rightarrow \Delta, \blacksquare B}{}
\hspace{5em}
$
\hspace{9em} $
\infer[J]{{\cal I} |\blacksquare B, \Pi \rightarrow \Theta}{}
\vspace{-.2em}
$\\ \vspace{-.5em}
\\
$\infer[cut]{{\cal H} |{\cal I} |\Gamma, \Pi \rightarrow \Delta, \Theta}{\hspace{28em}}$
}
\end{center}

\normalsize
There are four types of paths,
$Q_i, Q_j, R_k, R_l$, as above.
Note that, by {Facts 1.1, 1.2},
there is no application of the $\to$- nor $\Rightarrow$-rule
in $Q_i$ or $R_k$.

\vspace{1em}
$\bullet$ For $Q_i$,we are going to replace the lower hypersequent of $L_i$:
$\nabla \Lambda_i\to |
\rightarrow \blacksquare B| \Rightarrow$
with
${\cal I} | \nabla \Lambda_i\to |
\Pi \rightarrow \Theta| \Rightarrow$.

Use the subproof above $J$ as follows.
For the top $M_k$ of $R_k$,
utilize the $cut$-rule to delete $B$.
For the top $\blacksquare B\to \blacksquare B$
of $R_l$,
we can use the subproof above $L_i$.
Then, under these replacements,
we can construct the following proof.

\small
\begin{center}
\shortstack{
$
\infer[cut]
{
{\cal L}|
\nabla \Lambda_i\to | \Sigma_k \Rightarrow \Psi_k
}
{ \nabla \Lambda_i\to |
\Rightarrow B \hspace{1.5em}
{\cal L}|
B, \Sigma_k \Rightarrow \Psi_k
}
$
\hspace{2em}
\shortstack{
\hspace{-1.6em}
\deduc \\ \vspace{.5em} \\ \vspace{.5em}
$\infer[L_i]{\nabla \Lambda_i\to |
\rightarrow \blacksquare B| \Rightarrow}{}$
}
\\ \vspace{1em}
$R_k$ \hspace{1em} $\vdots$
\hspace{10em}
$R_l$ \hspace{1em} $\vdots$
\\
\hspace{2.6em}
$\vdots$
\hspace{12.3em} 
$\vdots$
\\ \vspace{1em}
\\
\hspace{2.6em}
$\ddots
\hspace{7.5em}\iddots
$
\vspace{2em}
\\
\hspace{2em}
$
\infer[I]{{\cal I} | \nabla \Lambda_i\to |
\Pi \rightarrow \Theta| \Rightarrow}{}
$
}
\end{center}

\normalsize

$\bullet$ For $Q_j$,
we just replace $L_j$ with the subproof above $J$
of the original proof.

\vspace{1em}
Then, under this replacement, simulate the subproof above $I$
to obtain the original end hypersequent.
as follows.

\small

\begin{center}
\shortstack{
$\hspace{1.8em}$
${\cal I} | \nabla \Lambda_i\to |
\Pi \rightarrow \Theta| \Rightarrow$
\hspace{5em}
\shortstack{
\deduc \\ \vspace{1em} \\ \vspace{1em}
$\infer[J]{{\cal I} |\blacksquare B, \Pi \rightarrow \Theta}{}$
}
\\ \vspace{1em}
$Q_i$ \hspace{1em} $\vdots$
\hspace{10em}
$Q_j$ \hspace{1em} $\vdots$
\\
\hspace{2.6em}
$\vdots$
\hspace{12.3em} 
$\vdots$
\\ \vspace{1em}
\\
\hspace{2.6em}
$\ddots
\hspace{7.5em}\iddots
$
\vspace{2em}
\\
\hspace{2em}
$
\infer[I]{{\cal H}|{\cal I} |
\Gamma, \Pi \rightarrow \Delta, \Theta}{}
$
}
\end{center}

\normalsize

\vspace{1em}
\noindent {\it Case 2.2}. The turnstile of $P$ is $\Rightarrow$.
The subproofs above $I$ and $J$ are separately depicted
due to the lack of space.
The subproof above $I$ is as follows.

\footnotesize

\begin{center}
\shortstack{
$\hspace{-2em}$
$\infer[L_i]{\nabla \Lambda_i\to |
\rightarrow \blacksquare B| \Rightarrow}{
\nabla \Lambda_i\to |
\Box B\Rightarrow B
}$
$\hspace{1.7em}$
$\infer[L_j]{\nabla \Lambda_j\to |
\Rightarrow \blacksquare B }{
\nabla \Lambda_j\to |
\Rightarrow B
}$
\hspace{4em}
$\infer[L_k]
{\blacksquare B \to \blacksquare B}{}
$
\\
\hspace{-4em}
$Q_i$ \hspace{1em} $\vdots$
\hspace{6em}
$Q_j$ \hspace{1em} $\vdots$
\hspace{8em}
$Q_k$ \hspace{1em} $\vdots$
\\
\hspace{-1.2em}
$\vdots$
\hspace{8.7em} 
$\vdots$
\hspace{10.7em} 
$\vdots$
\vspace{1em}
\\
\hspace{-.4em}
\infer[\Rightarrow]{\Xi_i\Rightarrow\Upsilon_i, \blacksquare B}{\Xi_i\rightarrow\Upsilon_i, \blacksquare B}
\hspace{5.3em}
\vdots
\hspace{8.3em}
\infer[\Rightarrow]{\Xi_k\Rightarrow\Upsilon_k, \blacksquare B}{\Xi_k\rightarrow\Upsilon_k, \blacksquare B}
\\ \vspace{1em}
\\
\hspace{-1.6em}
$\vdots$
\hspace{8.8em} 
$\vdots$
\hspace{10.7em} 
$\vdots$
\\ \vspace{1em}
\\
\hspace{-3em}
$\ddots
\hspace{5.5em}\vdots
$
\hspace{5.5em}
$\iddots
$
\vspace{2em}
\\
\hspace{2em}
$
\infer[I]{{\cal H} |\Gamma \Rightarrow \Delta, \blacksquare B}{}
\hspace{5em}
$
}
\end{center}

\normalsize
The subproof above $J$ is as follows.

\footnotesize
\begin{center}
\shortstack{
$\hspace{-2em}$
$\infer[M_l]
{{\cal L}|
\blacksquare B \to
| \Sigma_l \Rightarrow \Psi_l
}
{{\cal L}|
B, \Sigma_l \Rightarrow \Psi_l
}$
$\hspace{1.7em}$
$\infer[M_m]{ \nabla \Lambda_m\to |
\blacksquare B \Rightarrow
}{
\nabla \Lambda_m\to | B \Rightarrow
}$
\hspace{4em}
$\infer[M_n]
{\blacksquare B \to \blacksquare B}{}
$
\\
\hspace{-4em}
$R_l$ \hspace{1em} $\vdots$
\hspace{6em}
$R_m$ \hspace{1em} $\vdots$
\hspace{8em}
$R_n$ \hspace{1em} $\vdots$
\\
\hspace{-1.2em}
$\vdots$
\hspace{8.7em} 
$\vdots$
\hspace{10.7em} 
$\vdots$
\vspace{1em}
\\
\hspace{-.4em}
\infer[\Rightarrow]{\blacksquare B, \Xi_l\Rightarrow\Upsilon_l}{\blacksquare B, \Xi_l\rightarrow\Upsilon_l}
\hspace{5.8em}
\vdots
\hspace{8.8em}
\infer[\Rightarrow]{\blacksquare B, \Xi_n\Rightarrow\Upsilon_n}{\blacksquare B, \Xi_n\rightarrow\Upsilon_n}
\\ \vspace{1em}
\\
\hspace{-1.6em}
$\vdots$
\hspace{8.8em} 
$\vdots$
\hspace{10.7em} 
$\vdots$
\\ \vspace{1em}
\\
\hspace{-3em}
$\ddots
\hspace{5.5em}\vdots
$
\hspace{5.5em}
$\iddots
$
\vspace{2em}
\\
\hspace{2em}
$
\infer[J]{{\cal I} |\blacksquare B, \Pi \Rightarrow \Theta}{}
\hspace{5em}
$
}
\end{center}

\normalsize
There are six types of paths,
$Q_i, Q_j, Q_k, R_l, R_m, R_n$ in $P$ as above.
Note that by { Facts 1.1, 1.3},
there is no application of $\to$-rule and
exactly one application of $\Rightarrow$-rule
in $Q_i, Q_k, R_l$ and $R_n$;
by { Facts 2.1, 2.2},
there is no application of $\to$-rule or $\Rightarrow$-rule
in $Q_j$ and $R_m$.

\vspace{1em}

$\bullet$ For $Q_i$,
we are going to replace the lower hypersequent of the application of the $\Rightarrow$-rule in $Q_i$
with
$ {\cal I} |( \Pi \Rightarrow \Theta)\cdot (\Xi_i \Rightarrow \Upsilon_i)
$.
Making the replacement for the top of $R_l$, $R_m$ and $R_n$
we can construct the following
proof of it.

\footnotesize
\begin{center}
\shortstack{
$
\infer
{
{\cal L}|
\nabla \Lambda_{i} \to| \Sigma_l \Rightarrow \Psi_l
}
{ \nabla \Lambda_{i} \to |
\Rightarrow B \hspace{1.5em}
{\cal L}|
B, \Sigma_l \Rightarrow \Psi_l
}$
$\hspace{12em}$
$\hspace{5em}$
\hspace{4em}
\\ \vspace{1em}
\hspace{-8.6em}
$R_l$ \hspace{1em} $\vdots$
\hspace{10em}
\hspace{7em}
\\
\shortstack{
$\vdots$ \vspace{1em}
\\$\infer={
\nabla \Lambda_i \to
| \Xi_l \rightarrow \Upsilon_l
|\Rightarrow
}
{
\nabla \Lambda_i, \Xi_l \rightarrow \Upsilon_l
}$
}
\hspace{3em}
$
\infer={ \nabla \Lambda_m\to | \nabla
\Lambda_i \to | \Rightarrow}{
\nabla \Lambda_i\to |
\Rightarrow B \hspace{1.5em}
\nabla \Lambda_m\to | B \Rightarrow
}$
\hspace{3em}
\shortstack{
\deduc \vspace{1em} \\
$
{\Xi_i \to \Upsilon_i, \blacksquare B}{}
$
}
\\ \vspace{1em}
\hspace{-1.5em}
$Q_i$ \hspace{1em} $\vdots$
\hspace{9.5em}
$Q_i$ \hspace{1em} $\vdots$
\vspace{.5em}
\hspace{8em}
$R_n$ \hspace{1em} $\vdots$
\\
\hspace{2em}
$\infer{
\Xi_i, \Xi_l \Rightarrow \Upsilon_i, \Upsilon_l
}
{
\Xi_i, \Xi_l \rightarrow \Upsilon_i, \Upsilon_l
}
$
\hspace{5.5em}
$\infer={
\nabla \Lambda_m\to |
\Xi_i \Rightarrow \Upsilon_i
}{
\infer{ \nabla \Lambda_m,
\Xi_i \Rightarrow \Upsilon_i}
{ \nabla \Lambda_m,
\Xi_i \rightarrow \Upsilon_i}
}
$
\hspace{4em}
$\infer{\Xi_i, \Xi_n\Rightarrow\Upsilon_i, \Upsilon_n}{\Xi_i, \Xi_n\rightarrow\Upsilon_i, \Upsilon_n}
$
\\ \vspace{1em}
\hspace{-1.5em}
$R_l$ \hspace{1em} $\vdots$
\hspace{9.5em}
$R_m$ \hspace{1em} $\vdots$
\vspace{.5em}
\hspace{8em}
$R_n$ \hspace{1em} $\vdots$
\\ \vspace{.5em}
\hspace{1em}
$\ddots
\hspace{9.7em}\vdots
$
\hspace{7.5em}
$\iddots
$
\vspace{2em}
\\
\hspace{9em}
$
\infer[J]{{\cal I} |\Xi_i, \Pi \Rightarrow\Upsilon_i, \Theta}{}
\hspace{5em}
$
}
\end{center}

\normalsize
Here, the top of $R_n$, $\blacksquare B \to \blacksquare B$,
is replaced with the subproof above the application of the $\Rightarrow$-rule
in $Q_i$.

\vspace{1em}
$\bullet$ For $Q_j$,
we are going to replace
the lower hypersequent of
$L_j$
with
$\nabla \Lambda_j\to |
\Rightarrow
|
{\cal I} | \Pi \Rightarrow \Theta
$.

First, lift up
the application of the $\Rightarrow$-rule in $R_n$ to
the place just below $M_n$ as follows.

\small
\begin{center}
\shortstack{
$\infer[M_n]{\blacksquare B\to \blacksquare B}{}$
\\
\hspace{-4em}
$R_n$ \hspace{1em} $\vdots$
\\
\hspace{6.2em}$\vdots$ \hspace{6em} $\rhd$
\vspace{1em}
\\
$\infer[\Rightarrow]{\blacksquare B, \Xi_n \Rightarrow \Upsilon_n}{
\blacksquare B, \Xi_n \rightarrow \Upsilon_n
}$
\\
\hspace{-1em}
$\vdots$
}
\hspace{2em}
\shortstack{
$\infer[\Rightarrow]{\blacksquare B\Rightarrow \blacksquare B}{
\infer[M_n]{\blacksquare B\rightarrow \blacksquare B}{}
}$
\\
\hspace{-3em}
$R_n$ \hspace{1em} $\vdots$
\\
$\vdots$ \vspace{1.5em}
\\
$\infer[(\Rightarrow)]{\blacksquare B, \Xi_n \Rightarrow \Upsilon_n}{
}$
\\
\hspace{-1em}
$\vdots$
}
\end{center}

\normalsize
Here the transformation is possible
by the same argument as in the proof of Fact 2.1;
the part can be so transformed that there are only applications of
propositional rules and the rules of $iw, ic$,
and thus the turnstile in the part can be turned into $\Rightarrow$.

Then, using this transformation,
we can construct the following
proof.

\footnotesize
\begin{center}
\shortstack{
$\hspace{-2em}$
$
\infer
{
{\cal L}|
\nabla \Lambda_{j} \to | \Sigma_l \Rightarrow \Psi_l
}
{ \nabla \Lambda_{j} \to |
\Rightarrow B \hspace{1.5em}
{\cal L}|
B, \Sigma_l \Rightarrow \Psi_l
}$
$\hspace{1.7em}$
$
\infer{
\nabla \Lambda_{j} \to | \nabla \Lambda_m \to| \Rightarrow
}
{ \nabla \Lambda_{j} \to |
\Rightarrow B \hspace{1.5em}
\nabla \Lambda_{m} \to | B \Rightarrow
}
$
\hspace{2em}
$\infer[L_j]{\nabla \Lambda_j\to |
\Rightarrow \blacksquare B }{
\nabla \Lambda_j\to |
\Rightarrow B
}$
\\
\hspace{-4em}
$R_l$ \hspace{1em} $\vdots$
\hspace{10em}
$R_m$ \hspace{1em} $\vdots$
\hspace{8em}
$R_n$ \hspace{1em} $\vdots$
\\
\hspace{-1.2em}
$\vdots$
\hspace{12.8em} 
$\vdots$
\hspace{10.7em} 
$\vdots$
\vspace{1em}
\\
\hspace{1.3em}
$
\infer[\Rightarrow]{ \nabla \Lambda_{j} , \Xi_l\Rightarrow\Upsilon_l}{ \nabla \Lambda_{j} , \Xi_l\rightarrow\Upsilon_l}
\hspace{8.7em}
\vdots
\hspace{7.3em}
\infer[(\Rightarrow)]{\nabla \lambda_j \to | \Xi_n\Rightarrow\Upsilon_n}{
}
$
\\ \vspace{1em}
\\
\hspace{-1.6em}
$\vdots$
\hspace{13em} 
$\vdots$
\hspace{10.7em} 
$\vdots$
\\ \vspace{1.5em}
\\
\hspace{-1.4em}
$\ddots
\hspace{8.2em}\vdots
$
\hspace{5.5em}
$\iddots
$
\vspace{2em}
\\
\hspace{9em}
$
\infer=[4, 4^{\blacksquare},merge]{
{\cal I} |
\nabla \Lambda_j \to |
\Pi
\Rightarrow
\Theta 
}{
  \infer[J]{
  {\cal I} |
 \nabla \Lambda_j \to |
  \nabla \Lambda_j,
 \Pi
  \Rightarrow
  \Theta}{}
}
\hspace{5em}
$
}
\end{center}


\normalsize

Note that the rules of $4, 4^{\blacksquare}$ 
are applied serially
in the last part,
and it is the only place where these rules are used
for the cut-elimination.
\footnote{
{ 
This fact was pointed out by an anonymous referee.
}
}

$\bullet$ For $Q_k$,
firstly as done for $R_n$,
lift up
the application of the $\Rightarrow$-rule in $Q_k$ to
the place just below $L_k$.
Then,
we just replace $L_k$
with the subproof above $J$ of the original proof $P$.

\vspace{1em}
Under these replacements at the top of $Q_i, Q_j$ and $Q_k$,
we simulate the subproof above $I$
to obtain the original end hyperasequent as follows.

\footnotesize

\begin{center}
\shortstack{
$\hspace{6em}$
$\hspace{1.7em}$
${\cal I} |
\nabla \Lambda_j \to |
\Pi
\Rightarrow
\Theta $
\hspace{4em}
$\infer[J]{{\cal I} |\blacksquare B, \Pi \Rightarrow \Theta}{\deduc\vspace{1.5em}}
$
\\
\hspace{5.5em}
$Q_j$ \hspace{1em} $\vdots$
\hspace{8em}
$Q_k$ \hspace{1em} $\vdots$
\\
\hspace{8.4em} 
$\vdots$
\hspace{10.7em} 
$\vdots$
\vspace{1em}
\\
\hspace{-.4em}
\infer[(\Rightarrow)]{{\cal I} |\Xi_i, \Pi \Rightarrow\Upsilon_i, \Theta}{
}
\hspace{4.8em}
\vdots
\hspace{7.8em}
\infer[(\Rightarrow)]{
{\cal I} |
\Xi_k, \Pi\Rightarrow\Upsilon_k, \Theta}{
}
\\ \vspace{1em}
\\
\hspace{-4em}
$Q_i$ \hspace{1em} $\vdots$
\hspace{8.8em} 
$\vdots$
\hspace{10.7em} 
$\vdots$
\\ \vspace{2em}
\\
\hspace{-3em}
$\ddots
\hspace{5.5em}\vdots
$
\hspace{5.5em}
$\iddots
$
\vspace{2em}
\\
\hspace{2em}
$
\infer[I]{{\cal H} |{\cal I} |\Gamma, \Pi \Rightarrow \Delta, \Theta}{}
\hspace{5em}
$
}
\end{center}

\normalsize

\noindent {\it Case 3}. $A$ is an atomic formula $p$.
(We omit the other case: $A=\bot$.)

\noindent {\it Case 3.1}. The turnstile of $P$ is $\to$.
$P$ is as follows.

\begin{center}
$
\infer{{\cal H} | {\cal I} |
\Gamma, \Pi \to \Delta, \Theta
}{
\shortstack{ \vspace{1em}
\hspace{-1em}
$p\to p$
\\ \vspace{.5em}
\hspace{-3em}
\hspace{1em} \deduc
\\
\\ 
$\infer[I]{{\cal H} |
\Gamma \to \Delta, p}{}$
}
\hspace{4em}
\shortstack{
$p\to p$
\\ 
\hspace{-2.85em} $R_j$
\hspace{1em} $\vdots$
\\
$\vdots$
\vspace{1em}
\\ 
$\infer[J]{{\cal I} | p,
\Pi \to \Theta}{}$
}
}
$
\end{center}

By {Fact 1.2},
there is no application of the $\to$- or $\Rightarrow$- rule
in any $Q_i$ and any $R_j$.
For each $R_j$,
replace the top of $R_j$
with the subproof above $I$.
Then, under this replacement,
we simulate the subproof above $J$
to obtain the original end hypersequent.

\begin{center}
\shortstack{ \vspace{1em}
$p\to p$
\\ \vspace{.5em}
\hspace{-2em}
\hspace{1em} \deduc
\\ 
${\cal H} |
\Gamma \to \Delta, p$
\\ 
\hspace{-2.2em}
$R_j$ \hspace{1em} $\vdots$
\\
\hspace{.6em} 
$\vdots$
\vspace{1em}
\\ 
${\cal H} | {\cal I} |
\Gamma, \Pi \to \Delta, \Theta$
}
\end{center}

\noindent {\it Case 3.2}. The turnstile of $P$ is $\Rightarrow$.
$P$ is as follows.

\begin{center}
\shortstack{
$p \to p$
\hspace{6em}
$p \to p$
\\ 
\hspace{-2.5em}
$Q_i$ \hspace{1em} $\vdots$
\hspace{5.7em}
$R_j$ \hspace{1em} $\vdots$
\vspace{.5em}
\\
$\infer{
\Xi_i \Rightarrow \Upsilon_i, p
}
{
\Xi_i \rightarrow \Upsilon_i, p
}
$
\hspace{4em}
$\infer{
p, \Xi_j \Rightarrow \Upsilon_j
}
{
p, \Xi_j \rightarrow \Upsilon_j
}
$
\\ 
\hspace{-1em}
\hspace{1em} $\vdots$
\hspace{7em}
\hspace{1em} $\vdots$
\vspace{1em}
\\
$
\infer{{\cal H} | {\cal I} |\Gamma, \Pi
\Rightarrow
\Delta, \Theta}{
\infer[I]{{\cal H} |\Gamma
\Rightarrow \Delta, p }{}
\hspace{3em}
\infer[J]{{\cal I} | p, \Pi
\Rightarrow
\Theta}{}
}
$
}
\end{center}

By Facts 1.1, 1.3,
there is no application of the $\to$-rule
and there is exactly one application of the rule $\Rightarrow$-rule
in any $Q_i$ and any $R_j$.

For the lower hypersequent
$\Xi_i \Rightarrow \Upsilon_i, p$
of the $\Rightarrow$-rule in $Q_i$,
we are going to replace it with
${\cal I} |\Xi_i, \Pi \Rightarrow \Upsilon_i, 
\Theta$.
To do this,
replace the top of each $R_j$
with the subproof of
$\Xi_i \rightarrow \Upsilon_i, p$.
Then, under this replacement,
simulate the subproof above $J$ to obtain
$\Xi_i \Rightarrow \Upsilon_i|{\cal I} | \Pi
\Rightarrow
\Theta$.

\begin{center}
\shortstack{ 
\hspace{-1.2em}
$p\to p$
\\ \vspace{.5em}
\hspace{-4.3em}
$Q_i$ \hspace{1em} \vdots
\\ 
$
\Xi_i \rightarrow \Upsilon_i, p
$
\\ \vspace{.7em}
\hspace{-4em}
$R_j$ \hspace{1em} $\vdots$
\\
$\infer{
\Xi_i, \Xi_j \Rightarrow \Upsilon_i, \Upsilon_j
}
{
\Xi_i, \Xi_j \rightarrow \Upsilon_i, \Upsilon_j
}
$
\\ \vspace{.5em}
\hspace{-1em} \vdots
\\ \vspace{.5em}
$ {\cal I} | \Xi_i,
\Pi \to \Upsilon_i, \Theta$
}
\end{center}

Then, under this replacement for the lower hypersequent of the $\Rightarrow$-rule in each $Q_i$,
we can simulate the subproof above $I$
to obtain the original end hypersequent as follows.

\begin{center}
\shortstack{ 
$ {\cal I} | \Xi_i,
\Pi \Rightarrow \Upsilon_i, \Theta$
\\ 
\hspace{-2.2em}
$Q_i$ \hspace{1em} $\vdots$
\\
\hspace{.6em} 
$\vdots$
\vspace{1em}
\\ 
${\cal H} | {\cal I} |
\Gamma, \Pi \Rightarrow \Delta, \Theta$
}
\end{center}

This completes the proof of Theorem \ref{cut-elimination}.
\QED

\vspace{1em}
As a simple application of the cut-elimination theorem,
we verify of the conservability result of {\sf GR}
over {\sf GL}.

\begin{corollary}
Let $A^{\Box}$ be any formula of {\sf GR} containing no modality $\blacksquare$.
Then, the following equivalence holds.

\begin{center}
$A^{\Box}$ is provable in {\sf GR} $\Longleftrightarrow$
$A^{\Box}$ is provable in {\sf GL}.
\end{center}

\end{corollary}

\noindent {\bf Proof.}
'$\Longleftarrow$' is trivial.
On the direction of $'\Longrightarrow$',
suppose that $A^{\Box}$ is provable in {\sf GR}.
Then, by Theorem \ref{cut-elimination}, there is a cut-free proof of $A^{\Box}$ in {\sf GR},
where there is no application of the rules concerning $\blacksquare$.
When we simulated the inference rules of the hypersequent calculus {\sf GR}$_H$
in the axiomatic system of {\sf GR}$_{AX}$ in \S 2,
we used only the axioms of the axiomatic system {\sf GL}
for the inference rules not concerning $\blacksquare$.
Hence,
$A^{\Box}$ is provable in {\sf GL}.
\QED


\vspace{2em}
\large{\bf Technical Report 1}

\vspace{.5em}

\normalsize
Here we show a simple example of the transformation
of an application of  $\Box:r$ into the standard form.

\begin{center}
$\infer[I]{
 \rightarrow \blacksquare C |
 \blacksquare D\rightarrow |
 \Box E
   \rightarrow |
    \rightarrow \Box(D\wedge E)
}{
\infer[(\natural 2)]{
 \rightarrow \blacksquare C |
 \blacksquare D\rightarrow |
 \Box E
   \rightarrow |
    \Rightarrow^{\tau} D\wedge E
}{
\infer[(c)]{
 \rightarrow \blacksquare C |
 \blacksquare D\rightarrow |
  E  \Rightarrow^{\tau} D\wedge E}{
    \infer[(b)]{
 \rightarrow \blacksquare C |
 \Rightarrow^{\tau}|
 \blacksquare D\rightarrow |
 E
    \Rightarrow^{\tau} D\wedge E}{
 \infer[(a)]{
\Box C \Rightarrow  C |
 \blacksquare D\rightarrow |
 E
    \Rightarrow^{\tau} D\wedge E}{
     \infer[(\natural 2)]{
 \blacksquare D\rightarrow |
 E
    \Rightarrow^{\tau} D\wedge E}{
  \infer[(\natural 1)]{
D,
 E
    \Rightarrow^{\tau} D\wedge E}{
    \infer[Start \hspace{.3em}of  \hspace{.3em}
    \tau {\rm -}path
]{D\Rightarrow^{\tau} D}{D\to D}
    \hspace{1em}
        \infer[Start \hspace{.3em}of  \hspace{.3em}
    \tau {\rm -}path
]{E\Rightarrow^{\tau} E}{E\to E}
    }
    }
    }    
    }
    }
}
}
$
\end{center}

First we find an uppermost application
of ($\natural 1, \natural 2$)
such that there are applications of (a, b, c).
It is the application of $K$-rule of ($\natural 2$).
Permute it first with the one of (c)
then with the one of (b),
then with the one of (a).
Then, we obtain the following.

\begin{center}
$\infer[I]{
 \rightarrow \blacksquare C |
 \blacksquare D\rightarrow |
 \Box E
   \rightarrow |
    \rightarrow \Box(D\wedge E)
}{
\infer[(c)]{
 \rightarrow \blacksquare C |
 \blacksquare D\rightarrow |
 \Box E
   \rightarrow |  \Rightarrow^{\tau} D\wedge E}{
    \infer[(b)]{
 \rightarrow \blacksquare C |
 \Rightarrow^{\tau}|
 \blacksquare D\rightarrow |
 \Box E
   \rightarrow |
    \Rightarrow^{\tau} D\wedge E}{
 \infer[(a)]{
\Box C \Rightarrow  C |
 \blacksquare D\rightarrow |
  \Box E
   \rightarrow |
     \Rightarrow^{\tau} D\wedge E}{
     \infer[(\natural 2)]{
     \blacksquare D\rightarrow |
  \Box E
   \rightarrow |
     \Rightarrow^{\tau} D\wedge E
     }{
     \infer[(\natural 2)]{
 \blacksquare D\rightarrow |
 E
    \Rightarrow^{\tau} D\wedge E}{
  \infer[(\natural 1)]{
D,
 E
    \Rightarrow^{\tau} D\wedge E}{
    \infer[Start \hspace{.3em}of  \hspace{.3em}
    \tau {\rm -}path
]{D\Rightarrow^{\tau} D}{D\to D}
    \hspace{1em}
        \infer[Start \hspace{.3em}of  \hspace{.3em}
    \tau {\rm -}path
]{E\Rightarrow^{\tau} E}{E\to E}
    }
    }
    }    
    }
    }
}}
$
\end{center}

The next application to lift up is $I$.
First permute $I$ with the $merge$ of (c)
by duplicating $\Box:r$ to $I$ and $I'$.

\begin{center}
$\infer[(c)]{
 \rightarrow \blacksquare C |
 \blacksquare D\rightarrow |
 \Box E
   \rightarrow |
    \rightarrow \Box(D\wedge E)
}{
\infer=[I,I']{
 \rightarrow \blacksquare C |
\rightarrow \Box(D\wedge E) |
 \blacksquare D\rightarrow |
 \Box E
   \rightarrow |    \rightarrow \Box(D\wedge E)}{
    \infer[(b)]{
 \rightarrow \blacksquare C |
 \Rightarrow^{\tau'}|
 \blacksquare D\rightarrow |
 \Box E
   \rightarrow |
    \Rightarrow^{\tau} D\wedge E}{
 \infer[(a)]{
\Box C \Rightarrow  C |
 \blacksquare D\rightarrow |
  \Box E
   \rightarrow |
     \Rightarrow^{\tau} D\wedge E}{
     \infer[(\natural 2)]{
     \blacksquare D\rightarrow |
  \Box E
   \rightarrow |
     \Rightarrow^{\tau} D\wedge E
     }{
     \infer[(\natural 2)]{
 \blacksquare D\rightarrow |
 E
    \Rightarrow^{\tau} D\wedge E}{
  \infer[(\natural 1)]{
D,
 E
    \Rightarrow^{\tau} D\wedge E}{
    \infer[Start \hspace{.3em}of  \hspace{.3em}
    \tau {\rm -}path
]{D\Rightarrow^{\tau} D}{D\to D}
    \hspace{1em}
        \infer[Start \hspace{.3em}of  \hspace{.3em}
    \tau {\rm -}path
]{E\Rightarrow^{\tau} E}{E\to E}
    }
    }
    }    
    }
    }
}}
$
\end{center}

Here, we obtain another $\tau'$-path.
Then, permute $I$ with the application of (b)
and then with the one of (a).

\begin{center}
$\infer[(c)]{
 \rightarrow \blacksquare C |
 \blacksquare D\rightarrow |
 \Box E
   \rightarrow |
    \rightarrow \Box(D\wedge E)
}{
\infer[I']{
 \rightarrow \blacksquare C |
\rightarrow \Box(D\wedge E) |
 \blacksquare D\rightarrow |
 \Box E
   \rightarrow |    \rightarrow \Box(D\wedge E)}{
    \infer[(b)]{
 \rightarrow \blacksquare C |
 \Rightarrow^{\tau'}|
 \blacksquare D\rightarrow |
 \Box E
   \rightarrow |
    \rightarrow \Box(D\wedge E)}{
 \infer[(a)]{
\Box C \Rightarrow  C |
 \blacksquare D\rightarrow |
  \Box E
   \rightarrow |
      \rightarrow \Box(D\wedge E)}{
\infer[I]{
 \blacksquare D\rightarrow |
  \Box E
   \rightarrow |
      \rightarrow \Box(D\wedge E)}{
     \infer[(\natural 2)]{
     \blacksquare D\rightarrow |
  \Box E
   \rightarrow |
     \Rightarrow^{\tau} D\wedge E
     }{
     \infer[(\natural 2)]{
 \blacksquare D\rightarrow |
 E
    \Rightarrow^{\tau} D\wedge E}{
  \infer[(\natural 1)]{
D,
 E
    \Rightarrow^{\tau} D\wedge E}{
    \infer[Start \hspace{.3em}of  \hspace{.3em}
    \tau {\rm -}path
]{D\Rightarrow^{\tau} D}{D\to D}
    \hspace{1em}
        \infer[Start \hspace{.3em}of  \hspace{.3em}
    \tau {\rm -}path
]{E\Rightarrow^{\tau} E}{E\to E}
    }
    }
    }    
    }}
    }
}}
$
\end{center}

Now $I$ is of the standard form.
Also, in this case, $I'$ is already 
of the standard form
with the environment
$ \rightarrow \blacksquare C |
 \blacksquare D\rightarrow |
 \Box E
   \rightarrow |    \rightarrow \Box(D\wedge E)$.

\vspace{2em}
\large{\bf Technical Report 2}

\vspace{.5em}

Here we show a simple example of the top-down cut-elimination
procedure exposed in \S 3.3.
This example is concerned with Case 2.2 and Case 3.2 of the proof of Theorem 
\ref{cut-elimination}.

\begin{center}
    $
    \infer[cut]{
\Box p, \blacksquare q \to | \Rightarrow \Box(p\wedge q)
    }{
\infer[L_1=I]{
\Box p, \blacksquare q \to | \Rightarrow \blacksquare(p\wedge q)
}{
\infer=[K^\Box, K^\blacksquare]{\Box p, \blacksquare q \to | \Rightarrow p\wedge q}{
\infer{p, q \Rightarrow p\wedge q
}{
\infer{p, q \rightarrow p\wedge q
}{p\to p &q\to q}
}
}
}   
&
\infer[J]{ \blacksquare(p\wedge q)
\Rightarrow \Box(p\wedge q)
}{
\infer{\blacksquare(p\wedge q)
\rightarrow \Box(p\wedge q)}{
\infer{
\blacksquare(p\wedge q)\to |
\rightarrow\Box(p\wedge q)
}{
\infer[M_1]{
\blacksquare(p\wedge q)\to |
\Rightarrow p\wedge q
}{
\infer{p\wedge q
\Rightarrow p\wedge q}{
\infer={
p\wedge q, p\wedge q
\Rightarrow p\wedge q}{
\infer{p, q \Rightarrow p\wedge q}{
\infer{p, q \rightarrow p\wedge q
}{p\to p &q\to q}
}
}
}
}}
}
}
    }   
    $
\end{center}

By the top-down method described in Case 2.2,
we obtain the following proof.

\begin{center}
    $
    \infer=[4^\Box, 4^\blacksquare]{
    \Box p, \blacksquare q \to|
\Rightarrow \Box(p\wedge q)
    }{
\infer[J]{ \Box p, \blacksquare q 
\Rightarrow \Box(p\wedge q)
}{
\infer{
\Box p, \blacksquare q 
\rightarrow \Box(p\wedge q)}{
\infer{
\Box p, \blacksquare q \to |
\rightarrow\Box(p\wedge q)
}{
\infer[cut]{
\Box p, \blacksquare q \to |
\Rightarrow p\wedge q
}{
\infer=[K^\Box, K^\blacksquare]{\Box p, \blacksquare q \to | \Rightarrow p\wedge q}{
\infer{p, q \Rightarrow p\wedge q
}{
\infer{p, q \rightarrow p\wedge q
}{p\to p &q\to q}
}
}
&
\infer{p\wedge q
\Rightarrow p\wedge q}{
\infer={
p\wedge q, p\wedge q
\Rightarrow p\wedge q}{
\infer{p, q \Rightarrow p\wedge q}{
\infer{p, q \rightarrow p\wedge q
}{p\to p &q\to q}
}
}
}
}
}
}
}
 }   
    $
\end{center}

Now we focus on the proof above (and including) the new $cut$,
whose end hypersequent is $\Box p, \blacksquare q \to |
\Rightarrow p\wedge q$.
By the formula-reduction method,
we can easily make proofs of the following three.

\vspace{1em}
$\Box p, \blacksquare q \to |
\Rightarrow p$

$\Box p, \blacksquare q \to |
\Rightarrow q$

$p, q 
\Rightarrow p\wedge q$

\vspace{1em}
Using these proofs,
we can make cut-reductions as follows.

\small

\begin{center}
    $
    \infer{
      \Box p, \blacksquare q \to |
    \Rightarrow p\wedge q
    }{
 \infer[cut]{\Box p, \blacksquare q \to |
    \Box p, \blacksquare q \to |
    \Rightarrow p\wedge q
}{
\infer=[K^\Box, K^\blacksquare]{\Box p, \blacksquare q \to | \Rightarrow q}{
\infer{p, q \Rightarrow q
}{
\infer{p, q \rightarrow q
}{q\to q }
}
}
&
    \infer[cut]{
\Box p, \blacksquare q \to |
q \Rightarrow p\wedge q
}{
\infer=[K^\Box, K^\blacksquare]{\Box p, \blacksquare q \to | \Rightarrow p}{
\infer{p, q \Rightarrow p
}{
\infer{p, q \rightarrow p
}{p\to p }
}
}
&
\infer{p, q \Rightarrow p\wedge q}{
\infer{p, q \rightarrow p\wedge q
}{p\to p &q\to q}
}
}
}}
$
\end{center}

\normalsize
Then, we apply the top-down method described in {\it Case 3.2}.
First, we eliminate the upper $cut$, whose cut formula is $p$, as follows.


\begin{center}
    $
    \infer{
      \Box p, \blacksquare q \to |
    \Rightarrow p\wedge q
    }{
 \infer[cut]{\Box p, \blacksquare q \to |
    \Box p, \blacksquare q \to |
    \Rightarrow p\wedge q
}{
\infer=[K^\Box, K^\blacksquare]{\Box p, \blacksquare q \to | \Rightarrow q}{
\infer{p, q \Rightarrow q
}{
\infer{p, q \rightarrow q
}{q\to q }
}
}
&
    \infer=[K^\Box, K^\blacksquare]{
\Box p, \blacksquare q \to |
q \Rightarrow p\wedge q
}{
\infer={p, q, q  \Rightarrow p\wedge q}{
\infer{p, q, q \rightarrow p\wedge q
}{
\infer{
p, q\to p }{
p\to p
}
& q\to q}
}
}
}
}
$
\end{center}

\normalsize
Finally,  we can eliminate the remaining $cut$, whose cut formula is
$q$,
to obtain a cut-free proof of $ \Box p, \blacksquare q \to |
    \Rightarrow p\wedge q$
as follows.

\begin{center}
    $
    \infer{
      \Box p, \blacksquare q \to |
    \Rightarrow p\wedge q
    }{
\infer=[K^\Box, K^\blacksquare]{
\Box p, \blacksquare q \to |
\Box p, \blacksquare q \to |
 \Rightarrow p\wedge q
}{
    \infer=[K^\Box, K^\blacksquare]{
\Box p, \blacksquare q \to |
p,q  \Rightarrow p\wedge q
}{
\infer={p, q, p, q  \Rightarrow p\wedge q}{
\infer{p, q, p, q \rightarrow p\wedge q
}{
\infer{
p, q\to p }{
p\to p
}
& \infer{p, q \rightarrow q
}{q\to q }}
}
}
}}
$
\end{center}

\end{document}